\documentclass[11pt]{article}

\usepackage[sectionbib]{natbib}
\usepackage{booktabs}
\usepackage{lineno}
\usepackage{graphicx}
\usepackage{color,xcolor}
\usepackage{subfigure} 
\usepackage{subcaption}
\usepackage{geometry}
\usepackage{multirow}

\usepackage[colorlinks=true,linkcolor=blue,citecolor=blue]{hyperref}%

\usepackage{amsmath}
\usepackage{amssymb}
\usepackage{amsthm}

\usepackage[utf8]{inputenc} 
\usepackage{hyperref}       
\usepackage{url}            
\usepackage{booktabs}       
\usepackage{amsfonts}       
\usepackage{nicefrac}       
\usepackage{microtype}      
\usepackage[ruled,vlined]{algorithm2e}
\usepackage{authblk}
\usepackage{natbib}

\usepackage{array,epsfig,rotating}
\usepackage{tabularx}
\usepackage{subfigure}
\usepackage{graphicx}
\usepackage{verbatim}
\usepackage{arydshln}
\usepackage{authblk}

\newtheorem{theorem}{Theorem}
\newtheorem{lemma}{Lemma}

\newtheorem{definition}{Definition}

\newtheorem{assumption}{Assumption}

\allowdisplaybreaks

\def\hat{\widehat}

\newcommand{\bw}{\boldsymbol w}

\newcommand{\PP}{\mathbb P}

\newcommand{\bX}{\mathbf X}

\newcommand{\bH}{\mathbf H}

\newcommand{\bW}{\mathbf W}

\newcommand{\bSigma}{\boldsymbol \Sigma}
\newcommand{\btheta}{\boldsymbol \theta}

\newcommand{\E}{\mathbb E}

\newcommand{\Var}{\operatorname{Var}}

\def\spacingset#1{\renewcommand{\baselinestretch}%
{#1}\small\normalsize} \spacingset{1}
\spacingset{1}

\usepackage{float}

\usepackage{color-edits}

\begin{document}


\title{Sequentially-Rerandomized Switchback Experiments
}

\date{}

\author[1]{Zhenghao Zeng}
\author[2]{Christopher Adjaho}
\author[2]{Alonso Bucarey}
\author[1]{Chao Qin}
\author[2]{Ruixuan Zhang}
\author[2]{Paul Hoban}
\author[3]{Ramesh Johari}
\author[1]{Stefan Wager}

\affil[1]{Graduate School of Business, Stanford University}
\affil[2]{Airbnb}
\affil[3]{Management Science and Engineering Department, Stanford University}

\maketitle

\begin{abstract}
   Large-scale online platforms and marketplace systems often evaluate new policies through experiments that randomize treatment across operational units (e.g., geographies, regions, or clusters) over many time periods.
In these settings, standard A/B testing can be inefficient or unreliable due to a limited number of units, substantial cross-unit heterogeneity, non-stationarity, and potential carryover across periods.
We propose Sequentially-Rerandomized Switchback Experiments (SRSB), a new experimental design that helps mitigate these challenges.
SRSB re-randomizes treatment at each time period such as to enforce balance on pre-specified prognostic variables constructed from past observations. In the absence of carryover, SRSB improves precision by leveraging temporal dependence through balancing lagged outcomes and covariates; we develop finite-sample randomization inference under a sharp null as well as asymptotic inference as the number of periods grows. We then extend SRSB to settings with first-order carryover and introduce a blocked SRSB variant that rerandomizes within strata defined by the previous treatment to form stable and comparable “stay” groups. Extensive simulations demonstrate the practical gains and robustness of SRSB relative to standard switchback designs.
\end{abstract}

\section{Introduction}\label{sec:intro}

Companies often run large-scale experiments on online platforms and marketplace systems over time and across many operational units to evaluate business policies and guide product decisions \citep{vaver2011measuring, li2015toward, gupta2019top, ha2020counterfactual}. For example, advertisers may partition a market (e.g., a country) into geographic areas (``geos'') and randomize ad exposure at the geo level to estimate incremental revenue from advertising \citep{vaver2011measuring}. Similarly, ride-sharing platforms may assign drivers or regions in different geos to alternative dispatching policies to learn which policy improves efficiency and revenue. 

Designing and analyzing such experiments is challenging for several reasons. First, the number of experimental units can be small (e.g., tens to hundreds of geos), making asymptotic inference that relies on many units infeasible. Second, unit-level heterogeneity can be substantial. For instance, in geo experiments in France, \^{I}le-de-France (the Paris metropolitan area) is an outlier in population and economic activity, and imbalance on such a unit can affect both precision and interpretation. Third, the environment is dynamic and non-stationary: outcomes may exhibit strong seasonality \citep{ni2025enhancing}, persistent trends, or serial correlation. Finally, carryover effects across time periods may be present. For example, an ad campaign can have delayed or persistent impacts on future outcomes \citep{johnson2017online}. Together, these features render standard A/B testing ill-suited, or at best suboptimal, for the settings of interest.

To address these challenges, we propose an adaptive, design-based approach for experiments run over time and across multiple operational units. Specifically, we introduce \emph{Sequentially-Rerandomized Switchback Experiments} (SRSB), a procedure that uses outcomes and covariates observed up to time $t$ to adaptively construct the assignment at time $t$ by rerandomizing until a balance criterion is met. By balancing prognostic variables (notably lagged outcomes) at each period, SRSB aims to reduce estimation variance in dynamic environments where information such as past outcomes is predictive of future outcomes.

Because SRSB randomization depends on past outcomes, standard analytic strategies do not generally apply \citep{murphy2005experimental}.
Working in a finite-population framework in which potential outcomes and covariates are treated as fixed and all randomness arises from the treatment assignment, we study SRSB in two settings. In the absence of carryover effects, we show that sequential balancing reduces the variance of period-specific difference-in-means estimators and that the resulting gains accumulate over time; we develop two inference procedures, including exact randomization inference and asymptotic inference based on a martingale CLT as the number of time periods grows. We then extend SRSB to settings with first-order carryover effects and introduce a blocked variant tailored to this case, for which asymptotic normality is established by adapting techniques used to prove central limit theorems for mixingales \citep{mcleish1977invariance, davidson1992central}. Related results have also appeared independently in \citet{dempsey2025stable}, which establishes a central limit theorem for lagged martingale difference sequences using a similar ``Bernstein sums'' argument.

We demonstrate the practical performance of SRSB through extensive simulation studies, including semi-synthetic experiments calibrated to macroeconomic panel data (Penn World Table GDP) \citep{feenstra2015next, arkhangelsky2021synthetic} and an MDP-style carryover model that induces nonlinear, history-dependent dynamics. Across these settings, SRSB can yield substantial variance and RMSE reductions relative to complete randomization when lagged outcomes and covariates are predictive.

SRSB is motivated by two related literatures: switchback experiments and rerandomization. Switchback experiments provide a widely used alternative to standard A/B testing in dynamic settings \citep{bojinov2019time, bojinov2022, xiong2023bias, ni2023design, hu2024switchbackexperimentsgeometricmixing, ni2025enhancing, jia2025clusteredswitchbackdesignsexperimentation, missault2025robust, masoero2026multiple}. The key idea is to cluster individuals into operational units so that interference across units is negligible, and then switch each unit's treatment assignment over time. While much of the early literature focuses on the extreme case of a single unit, more recent work studies \emph{multi}-unit switchback experiments and allows for spillovers across units \citep{ni2023design, jia2025clusteredswitchbackdesignsexperimentation, missault2025robust, masoero2026multiple}. Most existing switchback designs rely on simple randomized assignment schemes (e.g., blockwise Bernoulli or complete randomization). Notably, \citet{xiong2023bias} and \citet{ni2025enhancing} leverage prior or pre-experiment data to improve efficiency and robustness, using empirical Bayes and time-series modeling, respectively, but do not adapt assignments based on outcomes observed during the experiment.

Rerandomization is a design-based strategy for improving precision by enforcing covariate balance \citep{morgan2012rerandomization, li2018asymptotic, zhou2018sequential, wang2023rerandomization, schindl2024unified}. It repeatedly draws candidate assignments and accepts only those that satisfy a pre-specified balance criterion, typically based on a Mahalanobis distance (or a related quadratic form) of treated--control mean differences. By discarding allocations that induce large imbalances, rerandomization yields more comparable treated and control groups and can reduce the variance of difference-in-means estimators when the balancing variables are prognostic for the potential outcomes. Extensions include rerandomization within strata \citep{wang2023rerandomization}, sequential variants for sequentially arriving units \citep{zhou2018sequential}, and more general acceptance rules based on alternative distances or balance metrics \citep{zhang2024pca, schindl2024unified}.

The most closely related works include sequential rerandomization \citep{zhou2018sequential} and rerandomization in switchback experiments \citep{ni2025reliable}. \citet{zhou2018sequential} studies rerandomization in settings where units arrive over time, assigning treatment by balancing covariates across all observed units. Our setting is sequential in a different sense: the set of units is fixed, but covariates and outcomes are revealed over time, and we use newly observed information to rerandomize and adaptively assign treatments across units at each period. \citet{ni2025reliable} studies rerandomization for single-unit switchback experiments when covariates for all periods are observed in advance, so that the entire treatment path can be selected prior to the experiment. In contrast, we consider multi-unit switchback experiments in which covariates and outcomes accrue sequentially, and rerandomization is performed sequentially to update cross-unit assignments period by period during the experiment. Our work is also related in spirit to synthetic control experiments \citep{abadie2021synthetic}. \citet{abadie2021synthetic} uses pre-treatment outcomes and covariates to construct synthetic treated and synthetic control aggregates and to select which aggregate units should be exposed to treatment. Their focus is on settings with one (or a few) treated units and a single intervention window, whereas our setting features repeated switching between treatment and control over time. In addition, their inferential approach relies on time-permutation arguments justified under stochastic assumptions on the outcome process, whereas we work in a finite-population framework in which the only source of randomness is the treatment assignment. 

The paper is organized as follows. In Section~\ref{sec:setup}, we introduce the multi-unit switchback experiment framework and the assumptions imposed on the potential outcomes, and we discuss the class of balancing variables used by the SRSB procedure. Section~\ref{sec:SRSB-nocarryover} then presents SRSB in settings without carryover effects and develops two inference approaches based on randomization inference and asymptotic theory. In Section~\ref{sec:SRSB-carryover}, we extend SRSB to settings with first-order carryover effects and establish the corresponding theoretical guarantees for the proposed estimator. Section~\ref{sec:simulation} reports extensive numerical experiments, including simulations on a semi-synthetic dataset, that illustrate the performance of our methods. All proofs and additional results are deferred to the supplementary materials.

\section{Setup \& Assumptions}\label{sec:setup}
In this section, we formalize the multi-unit switchback experiment setting and introduce the notation and assumptions used throughout the paper. We first describe the experiment and the potential outcomes framework, and then impose two assumptions on the set of potential outcomes to rule out interference and dependence of potential outcomes on future treatment assignments. Finally, we define the class of balancing variables used by SRSB, which motivates the SRSB procedure studied in subsequent sections.

\subsection{Multi-unit Switchback Experiments}

Suppose we are interested in estimating treatment effects on $N$ units over $T$ discrete time periods. We consider multi-unit switchback experiments: at each time $t$, each unit $i$ is assigned a binary treatment $W_{i,t} \in \{0,1\}$, after which an outcome $Y_{i,t}$ is observed. We collect the treatment trajectory for unit $i$ in the vector $\bW_{i,1:T}=(W_{i,1}, \dots, W_{i,T}) \in \{0,1\}^T$, and denote assignment across $N$ units at time $t$ by $\bW_{1:N,t}=(W_{1,t}, \dots, W_{N,t})^{\top} \in \{0,1\}^N$. Let $\bW=\{0,1\}^{N \times T}$ denote the treatment assignment matrix. We also use the lowercase counterparts $\bw_{i,1:T}$, $\bw_{1:N,t}$ and $\bw$ to denote fixed assignment realizations.  Throughout, we rely on the potential outcome framework \citep{rubin1974estimating} to define causal effects. In the most general case, the potential outcome for unit $i$ at time $t$ should depend on the matrix $\bw$, which allows for both spillover effects and carryover effects. However, in an experiment, we observe only a single realized treatment assignment matrix and the corresponding outcomes. The remaining potential outcomes are counterfactual and therefore unobserved, which makes treatment effect estimation challenging. To simplify the problem, we impose the following two assumptions on the potential outcomes.


\begin{assumption}[Non-anticipatory outcomes]\label{assumption:non-anticipating}
For each unit $i \in [N]$ and time $t \in [T]$, the potential outcome $Y_{i,t}(\bw)$ depends only on the treatment history up to time $t$. That is, for any two assignment paths $\bw$ and $\bw'$ satisfying $\bw_{1:N,1:t} = \bw_{1:N,1:t}'$, where $\bw_{1:N,1:t}=(\bw_{1:N,1}, \dots, \bw_{1:N,t}) \in \mathbb{R}^{N \times t}$, we have
\[
Y_{i,t}(\bw) = Y_{i,t}(\bw') = Y_{i,t}(\bw_{1:N,1:t}).
\]
\end{assumption}

\begin{assumption}[No spillover effects]\label{assumption:no-spillover}
    For each unit $i$ and time $t$, the potential outcome $Y_{i,t}(\bw)$ only depends on the treatment trajectory of unit $i$. That is, for any two assignment matrix $\bw$ and $\bw'$ satisfying $\bw_{i,1:T}=\bw_{i,1:T}'$, we have
    \[
    Y_{i,t}(\bw) = Y_{i,t}(\bw')=Y_{i,t}(\bw_{i,1:T}).
    \]
\end{assumption}

Assumption~\ref{assumption:non-anticipating} is commonly adopted in the literature \citep{bojinov2019time, rambachan2019econometric, bojinov2022, ni2025enhancing} and rules out anticipatory effects: the potential outcomes at time $t$ do not depend on future treatment assignments. In many applications, the future assignment schedule is not revealed to the units (e.g., drivers on ride-sharing platforms or regions in geo-level experiments), so units’ behavior at time $t$ cannot depend on treatments scheduled for later periods. 

Assumption~\ref{assumption:no-spillover} requires that each unit's outcomes only depend on its own treatment assignment trajectory, known as the stable-unit-treatment-value (SUTVA) assumption \citep{imbens2015causal,missault2025robust}. This assumption can fail when units interact. In switchback experiments, however, practitioners typically define operational units by clustering individuals so that cross-unit interference is negligible, in which case Assumption~\ref{assumption:no-spillover} is a reasonable approximation.

Under Assumption \ref{assumption:non-anticipating}--\ref{assumption:no-spillover},
the potential outcome of unit $i$ at time $t$ only depends on the treatment trajectory $\bw_{i,1:t}$ and 
we denote it as $Y_{i,t}(\bw_{i,1:t})$. The collection of all potential outcomes is
\[
\mathcal{Y} = \left\{Y_{i,t}(\bw_{i,1:t}) : \bw_{i,1:t} \in \{0,1\}^t,\ 1\leq t \leq T, \, 1\leq i \leq N \right\}.
\]
We adopt the design-based perspective \citep{neyman1990on, imbens2015causal, abadie2020sampling} and treat $\mathcal{Y}$ as fixed but unknown, so that all randomness arises solely from the treatment assignment. 

These assumptions substantially simplify the problem by reducing the number of distinct potential outcomes. The observed outcome for unit $i$ at time $t$ in an experiment with treatment $\bW_{i,1:T}$ is then $Y_{i,t}=Y_{i,t}(\bW_{i,1:t})$. The following assumption imposes additional structure on the potential outcomes depending on whether carryover effects are present to further reduce the number of potential outcomes in $\mathcal{Y}$. 
\begin{assumption}\label{assumption:carryover}
    \begin{enumerate}
        \item[(a)] \textbf{No carryover effect.}
For each unit $i \in [N]$ and time $t \in [T]$, the potential outcome at time $t$ depends on the treatment history only through the treatment at time $t$. Equivalently, for any two assignment paths $\bw_{i,1:T}$ and $\bw_{i,1:T}'$ such that $w_{i,t}=w_{i,t}'$, we have
\[
Y_{i,t}(\bw_{i,1:T}) = Y_{i,t}(\bw_{i,1:T}') = Y_{i,t}(w_{i,t}).
\]
    \item[(b)] \textbf{First-order carryover effect.}
    For each unit $i \in [N]$ and each time $t \in [T]$ with $t\ge 2$, the potential outcome at time $t$ depends on the treatment history only through the treatments at times $t-1$ and $t$. Equivalently, for any two assignment paths $\bw_{i,1:T}$ and $\bw'_{i,1:T}$ satisfying $w_{i,t}=w'_{i,t}$ and $w_{i,t-1}=w'_{i,t-1}$, we have
\[
Y_{i,t}(\bw_{i,1:T}) = Y_{i,t}(\bw'_{i,1:T}) = Y_{i,t}(w_{i,t-1},w_{i,t}).
\]
    \end{enumerate}
\end{assumption}
Assumption~\ref{assumption:carryover}(a) rules out carryover effects, whereas Assumption~\ref{assumption:carryover}(b) permits first-order carryover, meaning that outcomes may depend on the lag-one treatment. These two assumptions are both special cases of finite-order carryover effects \citep{bojinov2022}. We study design and estimation under these two settings in Sections~\ref{sec:SRSB-nocarryover} and~\ref{sec:SRSB-carryover}, respectively.

\subsection{Sequentially-Rerandomized Switchback Experiments}
In this work we focus on a design we call the Sequentially-Rerandomized Switchback Experiment (SRSB).  The SRSB assigns treatment at each time $t$ adaptively, using information available up to time $t$, in order to estimate treatment effects more efficiently. Let $\bX_{i,t}\in\mathbb{R}^{d_x}$ denote the (fixed) covariate vector for unit $i$ at time $t$, assumed to be observed prior to assignment at time $t$. The key idea of SRSB is that, for each $t\ge 2$, we repeatedly draw candidate assignments and accept the first that attains adequate balance on a pre-specified set of balancing variables $\bH_{i,t}\in\mathbb{R}^d$; the accepted assignment is then implemented as $\bW_{1:N,t}$. For example, one may take $\bH_{i,t}=(\bX_{i,t},\,Y_{i,t-1})^\top$ to balance contemporaneous covariates and lagged outcomes. More generally, we allow $\bH_{i,t}$ to be any function of the information observed immediately before assigning the treatments at time $t$: for each $t\in[T]$, we consider balancing variables of the form
\begin{equation}\label{eq:balancing-set}
\bH_{i,t}
=
f_t\!\left(\bX_{i,1},\dots,\bX_{i,t},\, Y_{i,1},\dots,Y_{i,t-1}\right),
\qquad \forall i\in[N].
\end{equation}
In practice, the choice of $\bH_{i,t}$ is guided by domain knowledge about which variables are prognostic for outcomes at time $t$. Intuitively, balancing covariates and past outcomes that are predictive of the potential outcomes at time $t$ can reduce estimation variance and improve the accuracy of treatment effect estimation.



Recall that in our design-based framework the potential outcomes and covariates are treated as fixed, and the only source of randomness is the treatment assignment. Assumptions~\ref{assumption:non-anticipating} and~\eqref{eq:balancing-set} imply that $\bH_{i,t}$ is measurable with respect to the filtration generated by past assignments, i.e.,
\[
\bH_{i,t} \in \sigma\!\left(\bW_{1:N,1},\dots,\bW_{1:N,t-1}\right) =: \mathcal{F}_{t-1}, \qquad \forall i \in [N].
\]
Accordingly, the observed trajectory for unit $i$ takes the form
\[
\bX_{i,1}, W_{i,1}, Y_{i,1},\ \bX_{i,2}, W_{i,2}, Y_{i,2},\ \dots,\ \bX_{i,T}, W_{i,T}, Y_{i,T},
\]
where each $\bW_{1:N,t}$ is selected by balancing $\{\bH_{i,t} : 1 \le i \le N\}$. We describe the SRSB procedure in detail and study its theoretical properties in the following sections.

\section{SRSB Without Carryover Effects}\label{sec:SRSB-nocarryover}

In this section, we assume that Assumption~\ref{assumption:carryover}(a) holds, so that carryover effects are absent. We then introduce the SRSB procedure and study estimation and inference of treatment effects in this no-carryover setting.


Under Assumption~\ref{assumption:carryover}(a), the observed outcome for unit $i$ at time $t$ depends only on the current assignment $W_{i,t}$, and we may write $Y_{i,t}=Y_{i,t}(W_{i,t})$. We consider the sample average treatment effect (SATE) averaged over units and time periods,
\[
\bar{\tau} \;=\; \frac{1}{NT}\sum_{i=1}^N\sum_{t=1}^T \bigl\{Y_{i,t}(1)-Y_{i,t}(0)\bigr\}.
\]
And we also denote the treatment effect at time $t$ as 
\[
\tau_t = \frac{1}{N}\sum_{i=1}^N \bigl\{Y_{i,t}(1)-Y_{i,t}(0)\bigr\}.
\]
Throughout, we adopt a finite-population perspective in which all potential outcomes are treated as fixed, so that the estimand $\bar{\tau}$ is well-defined. We now present the SRSB procedure for this no-carryover setting.

\subsection{Sequential Rerandomization}

The key idea of sequential rerandomization is to enforce balance \emph{period by period}. For each time $t$, we repeatedly draw a candidate assignment vector $\bW^*_{1:N,t} = (W^*_{1,t},\dots,W^*_{N,t})^\top$ and evaluate the resulting imbalance in the balancing variables. Specifically, for the balancing variables $\bH_{i,t}\in\mathbb{R}^d$, we compute the treated--control difference in their sample mean vectors under $\bW^*_{1:N,t}$,
\[
\hat{\btheta}_t^* \;=\; \frac{2}{N}\sum_{i=1}^N W_{i,t}^*\,\bH_{i,t}
\;-\;
\frac{2}{N}\sum_{i=1}^N (1-W_{i,t}^*)\,\bH_{i,t}.
\]
We accept $\bW^*_{1:N,t}$ as the realized assignment $\bW_{1:N,t}$ if the corresponding Mahalanobis distance (defined in Algorithm~\ref{alg:srsb}) is below a pre-specified threshold; otherwise we discard $\bW^*_{1:N,t}$ and continue rerandomizing until an acceptable assignment is found (or until the maximum number of draws is reached). The full procedure is summarized in Algorithm~\ref{alg:srsb}.

In Algorithm~\ref{alg:srsb}, we use the Mahalanobis distance to quantify covariate imbalance. Other choices are also possible, including the $\ell_2$ distance or, more generally, distances induced by quadratic forms; see \citet{schindl2024unified} for a unified perspective. Intuitively, rerandomization eliminates allocations that induce large imbalances in $\bH_{i,t}$. When $\bH_{i,t}$ is predictive of the potential outcomes $\{Y_{i,t}(1),Y_{i,t}(0)\}$, enforcing balance on $\bH_{i,t}$ reduces the component of outcome variation that can be explained by $\bH_{i,t}$, thereby improving precision \citep{morgan2012rerandomization, li2018asymptotic}. In sequential experiments, lagged outcomes and contemporaneous covariates are often highly predictive of future outcomes, so a natural choice is to include $Y_{i,t-1}$ and $\bX_{i,t}$ in $\bH_{i,t}$. More generally, one may incorporate additional lags (e.g., $Y_{i,t-2}$ and $\bX_{i,t-1}$) when they are expected to be informative based on domain knowledge. Under the accepted assignment $\bW_{1:N,t}$, the treated and control groups are comparable in the sense that their difference in the balancing variables is small, and each group is (approximately) representative of the underlying finite population of $N$ units. 

After implementing $\bW_{1:N,t}$, we observe $Y_{1,t},\dots,Y_{N,t}$ and proceed to time $t+1$, where the balancing variables $\bH_{i,t+1}$ are constructed from information available up to time $t+1$. As a result, the assignment sequence $\{\bW_{1:N,t},\, 1 \leq t \leq T\}$ is generally dependent across time: $\bW_{1:N,t+1}$ is chosen using the realized outcomes at time $t$, which themselves depend on $\bW_{1:N,t}$.

\begin{algorithm}[p]
\caption{Sequentially-Rerandomized Switchback Experiments (SRSB)}\label{alg:srsb}
\KwIn{Threshold $c>0$ (e.g., $c=\chi^2_d(\alpha)$) and maximum number of candidate assignments $n\in\mathbb{N}$.}
\KwOut{Treatment assignments $\bW_{1:N,1:T} \in \{0,1\}^{N\times T}$ and observed outcomes $\{Y_{i,t}, 1\leq i \leq N, 1\leq t \leq T\}$.}

\BlankLine

\For{$t=1,\ldots,T$}{
\textbf{Construct balancing variables:} Choose $\bH_{i,t}\in\mathbb{R}^d$ for each $i\in[N]$ based only on information available up to assigning treatments at time $t$ \;

Compute
\[
\overline{\bH}_t \gets \frac{1}{N}\sum_{i=1}^N \bH_{i,t}, 
\qquad
\hat{\bSigma}_t \gets \frac{4}{N^2}\sum_{i=1}^N (\bH_{i,t}-\overline{\bH}_t)(\bH_{i,t}-\overline{\bH}_t)^\top .
\]
Initialize $m\gets 0$, $d_{\min}\gets +\infty$, and $\bW_{1:N,t}^{\min}\gets \mathbf{0}$\;

\While{$m < n$}{
Sample a candidate assignment $\bW_{1:N,t}^* = (W_{1,t}^*,\ldots,W_{N,t}^*)^\top$ with exactly $N/2$ treated units\;
Compute the imbalance vector
\[
\hat{\btheta}_t^* \gets \frac{2}{N}\sum_{i=1}^N W_{i,t}^* \bH_{i,t}
\;-\;
\frac{2}{N}\sum_{i=1}^N (1-W_{i,t}^*) \bH_{i,t}.
\]
Compute the distance
\[
d^* \gets (\hat{\btheta}_t^*)^\top \hat{\bSigma}_t^{-1}\hat{\btheta}_t^* .
\]
\If{$d^* < d_{\min}$}{
$d_{\min}\gets d^*$ and $\bW_{1:N,t}^{\min}\gets \bW_{1:N,t}^*$\;
}
\If{$d^* < c$}{
Set $\bW_{1:N,t} \gets \bW_{1:N,t}^*$ and \textbf{break}\;
}
$m \gets m+1$\;
}
\If{$m = n$}{
Set $\bW_{1:N,t} \gets \bW_{1:N,t}^{\min}$\;
}
\textbf{Implement and observe:} Implement $\bW_{1:N,t}$ and observe $\{Y_{i,t}:1\le i\le N\}$\;
}
\Return{$\bW_{1:N,1:T}$} and $\{Y_{i,t}, 1\leq i \leq N, 1\leq t \leq T\}$\;
\end{algorithm}

Simple simulation-based illustrations are provided in Figure~\ref{fig:out-traj}. In this simple example, the control potential outcomes, $Y_{i,t}(0)$, are generated from an AR(1) model with covariates, as detailed in Section~\ref{sec:simu-simple-nocarryover}. The treated potential outcomes are then defined as $Y_{i,t}(1) = Y_{i,t}(0) + 1$, so that the target treatment effect is equal to $1$. The top panels show average outcome trajectories computed separately for groups defined by the treatment assignment at time $t=1$, under a design that uses completely randomized assignment within each time period. Three panels correspond to three independent realizations. Note that we are actually focused on the time period $t=6$ and plot the periods $t=1,\dots,6$. But for clarity we shift the horizon so that the time period focused is $t=1$. The difference between the green and orange lines at $t=1$ can be interpreted as the treatment effect estimator.
In the top panel, all realized assignments at time $t=1$ fail to balance the previously observed outcomes. In the first two panels, the difference between the group-average trajectories at $t=1$ is positive and grows over time, leading to an overestimate of the treatment effect at $t=1$. In the third panel, the difference is negative, leading to an underestimate at $t=1$; notably, the estimated treatment effect is even negative despite the true effect being $1$. These examples illustrate that failing to balance lagged outcomes and trends can induce substantial estimation error.

Similarly, the bottom panels in Figure~\ref{fig:out-traj} plot the average outcome trajectories under rerandomization that balances all previously observed outcomes. By enforcing balance on lagged outcomes and time trends, rerandomization reduces the error induced by treated--control imbalance; accordingly, the resulting estimator (the difference between the two lines at $t=1$) is close to $1$. Compared with the top panels, the estimate is substantially more accurate.

\begin{figure}[t]
	\centering
	\subfigure{
		\begin{minipage}[t]{0.32\linewidth}
			\centering
			\includegraphics[width=2.1in]{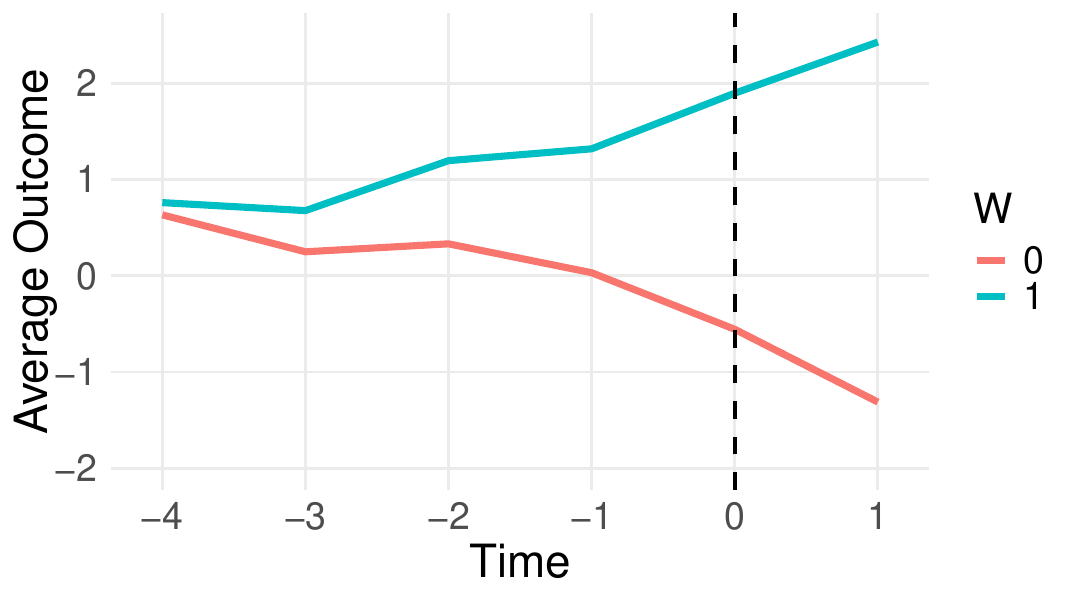}
	\end{minipage}}
	\subfigure{
		\begin{minipage}[t]{0.32\linewidth}
			\centering
			\includegraphics[width=2.1in]{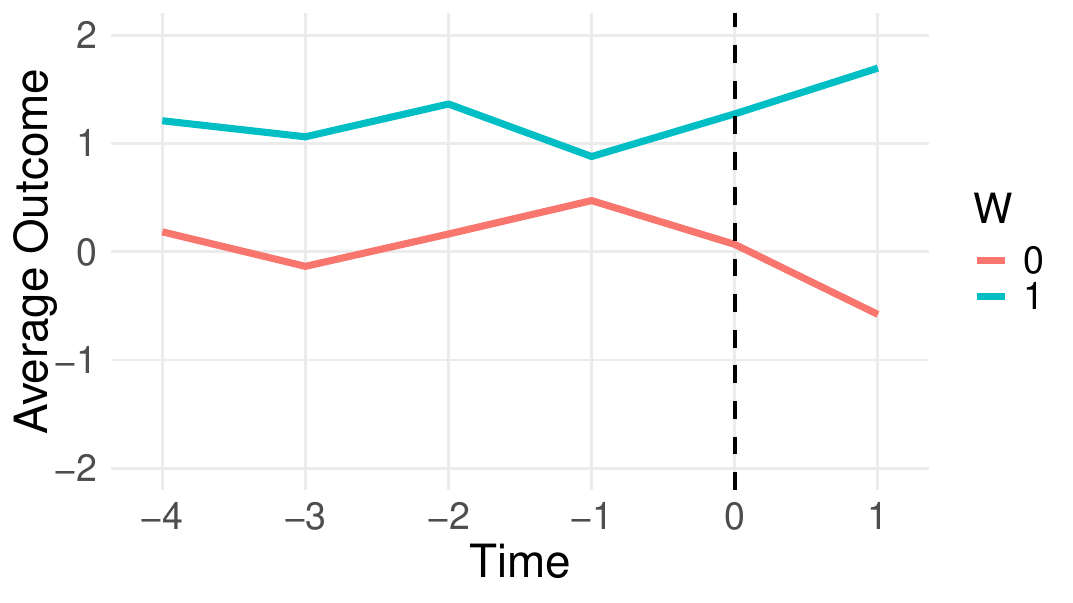}
	\end{minipage}}
    \subfigure{
		\begin{minipage}[t]{0.32\linewidth}
			\centering
			\includegraphics[width=2.1in]{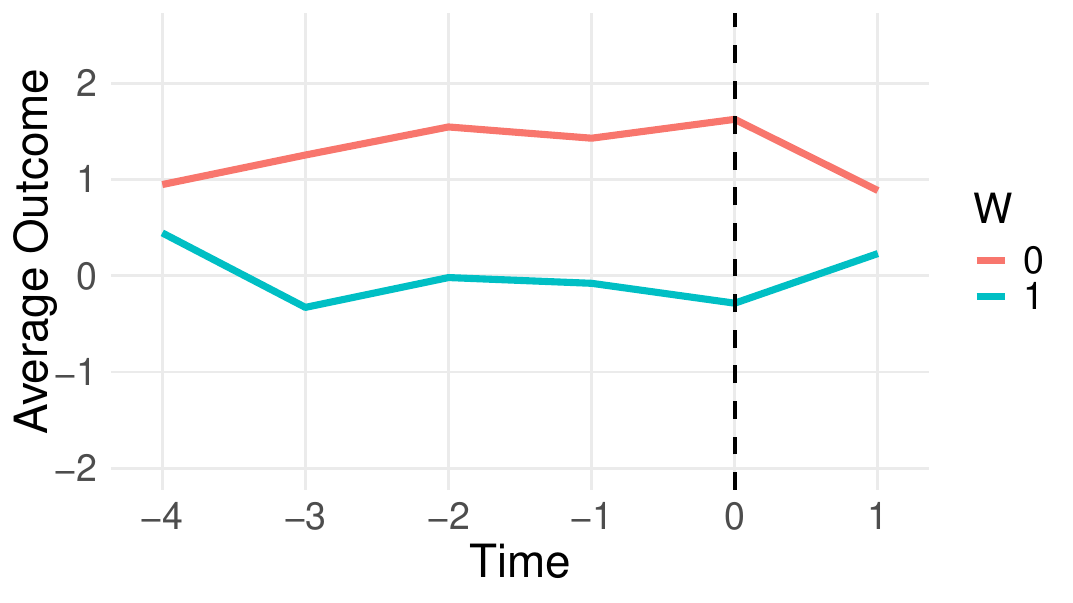}
	\end{minipage}}\\
    \subfigure{
		\begin{minipage}[t]{0.32\linewidth}
			\centering
			\includegraphics[width=2.1in]{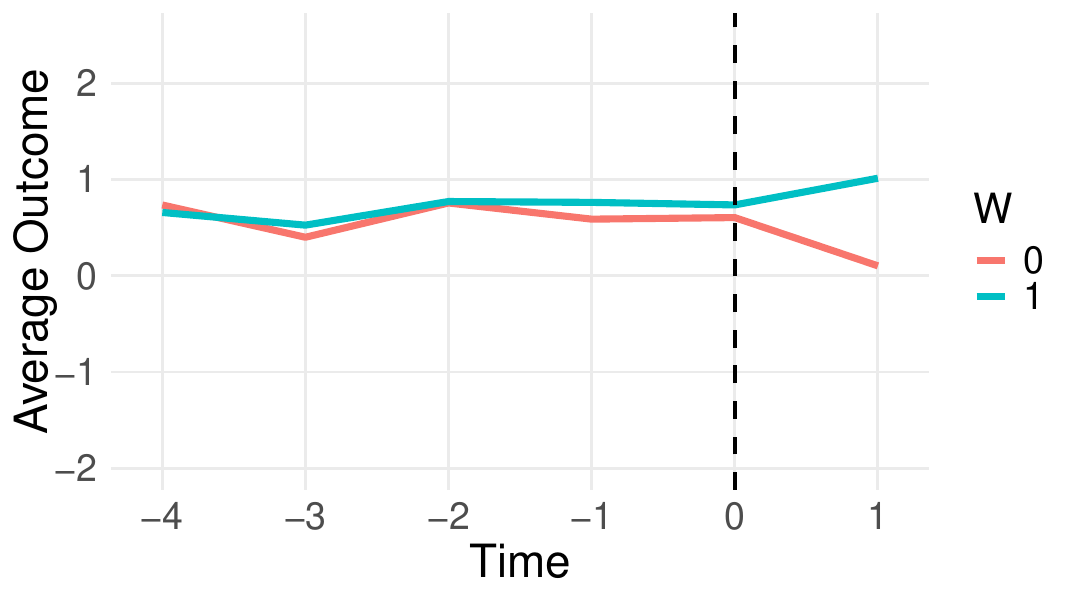}
	\end{minipage}}
	\subfigure{
		\begin{minipage}[t]{0.32\linewidth}
			\centering
			\includegraphics[width=2.1in]{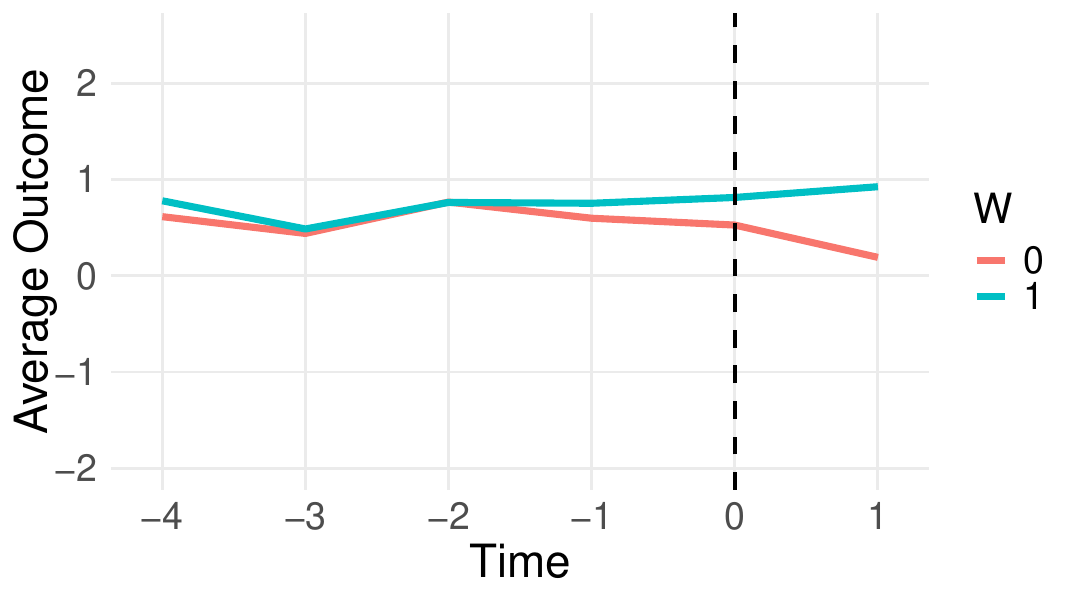}
	\end{minipage}}
    \subfigure{
		\begin{minipage}[t]{0.32\linewidth}
			\centering
			\includegraphics[width=2.1in]{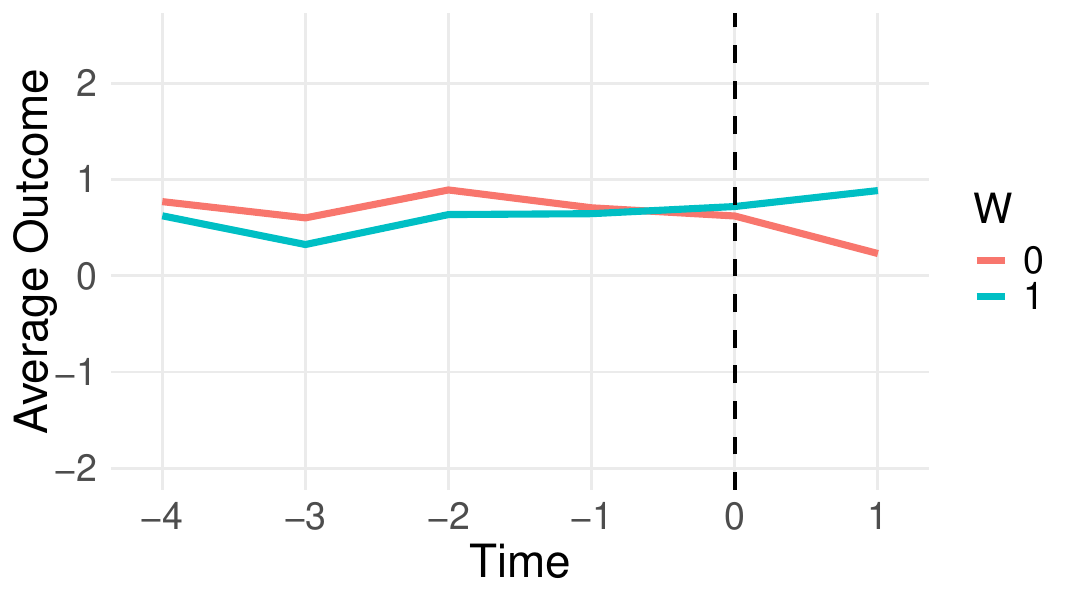}
	\end{minipage}}
	\centering
	\caption{Results from a toy simulation example where the potential outcomes are generated from an AR(1) model with covariates; see Section~\ref{sec:simu-simple-nocarryover} for details of the data-generating process. \textbf{Top:} Examples of average outcome trajectories under complete randomization at each time period. \textbf{Bottom:} Examples of average outcome trajectories under rerandomization that balances all previously observed outcomes in the same setting. The true treatment effect is \(1\). At each time point, average outcomes are computed separately for groups defined by the treatment assignment at time \(t=1\). The vertical dashed line at \(t=0\) marks the decision time: outcomes up to \(t=0\) are observed before assigning treatment at \(t=1\). In the top panel, the difference-in-means estimator at \(t=1\) can substantially overestimate (first two examples) or underestimate (third example) the treatment effect due to imbalance in lagged outcomes and pre-treatment trends. In the bottom panel, rerandomization reduces such imbalance by enforcing balance on lagged outcomes, leading to more accurate estimation of the treatment effect at \(t=1\).}
	\label{fig:out-traj}
\end{figure}

\subsection{Estimation and Inference}

In the remainder of this section, we develop estimation and inference procedures for the SATE $\bar{\tau}$ under the SRSB design.
In the no-carryover setting, we estimate the contemporaneous treatment effect separately at each time period and then aggregate these estimates across time. Specifically, let
\[
\hat{\tau}_t
\;=\;
\frac{2}{N}\sum_{i=1}^N W_{i,t}Y_{i,t}
\;-\;
\frac{2}{N}\sum_{i=1}^N (1-W_{i,t})Y_{i,t},
\]
denote the standard difference-in-means estimator at time $t$, where the scaling corresponds to complete randomization with $N/2$ treated units. We then average the period-specific contrasts to obtain an estimator of the SATE,
\begin{equation}\label{eq:est-nocarryover}
\hat{\tau}
\;=\;
\frac{1}{T}\sum_{t=1}^T \hat{\tau}_t .
\end{equation}
While the assignment sequence $\{\bW_{1:N,t}, 1\leq t \leq T\}$ can exhibit serial dependence under SRSB because $\bW_{1:N,t}$ is constructed using past outcomes, the estimator in \eqref{eq:est-nocarryover} remains an unbiased aggregation of within-period treated--control contrasts. In what follows, we present two approaches for conducting statistical inference for $\bar{\tau}$ based on $\hat{\tau}$ under the SRSB design.

\subsubsection{Randomization Inference}

The first approach is \emph{randomization-based inference} \citep{fisher1925design, neyman1990on, rosenbaum2002observational, ding2016randomization,ritzwoller2024randomization}. Fix $\delta\in\mathbb{R}$ and consider testing the sharp null hypothesis
\[
H_0(\delta):\quad Y_{i,t}(1)=Y_{i,t}(0)+\delta,\qquad \forall\, i\in[N],\ t\in[T],
\]
which posits a constant additive effect across units and time periods. Under this sharp null, the entire collection of potential outcomes
\[
\mathcal{Y}=\{Y_{i,t}(0),Y_{i,t}(1): 1\le i\le N,\ 1\le t\le T\}
\]
is point-identified from the observed outcomes, so the SRSB acceptance rule induces a well-defined randomization distribution over treatment assignment paths. The observed assignment trajectory $\bW_{1:N,1:T}$ can therefore be viewed as a single draw from this randomization distribution. Randomization inference approximates this distribution by Monte Carlo: we repeatedly generate assignment paths that satisfy the SRSB balancing criterion under $H_0(\delta)$, recompute the test statistic for each path, and obtain a $p$-value by comparing these simulated statistics to the observed one. The full procedure is summarized in Algorithm~\ref{alg:ri-SRSB}.

An important advantage of randomization inference is its finite-sample validity: for any fixed $(N,T)$, the resulting $p$-value is valid under the sharp null. This can be particularly appealing when the experiment involves only a small number of units and time periods. A confidence set for $\bar{\tau}$ can be constructed by inverting the test, i.e., collecting values of $\delta$ that are not rejected at a given significance level. The main drawback is computational cost, since the method relies on repeatedly simulating SRSB-compatible assignment paths. In addition, the sharp null is imposed on individual treatment effects. Consequently, rejecting $H_0(0)$ indicates the presence of treatment effects for at least some $(i,t)$, but does not by itself imply that the average effect $\bar{\tau}$ is nonzero.

\begin{algorithm}[p]
\caption{Randomization inference for the sharp null $H_0(\delta):\ Y_{i,t}(1)-Y_{i,t}(0)=\delta$}\label{alg:ri-SRSB}
\KwIn{Observed data $\{(\bX_{i,t},W_{i,t},Y_{i,t}) : i\in[N],\, t\in[T]\}$, sharp-null effect $\delta\in\mathbb{R}$, SRSB threshold $c>0$, number of Monte Carlo draws $M\in\mathbb{N}$.}
\KwOut{$p$-value for testing $H_0$.}

\BlankLine
\textbf{Step 0 (Observed statistic).}
Compute
\[
\hat{\tau}_t
\gets
\frac{2}{N}\sum_{i=1}^N W_{i,t}Y_{i,t}
-\frac{2}{N}\sum_{i=1}^N (1-W_{i,t})Y_{i,t},
\qquad
\hat{\tau}\gets \frac{1}{T}\sum_{t=1}^T \hat{\tau}_t .
\]

\BlankLine
\textbf{Step 1 (Impute potential outcomes under $H_0$).}
For all $i\in[N]$, $t\in[T]$, set
\[
Y_{i,t}(0)\gets Y_{i,t}-\delta W_{i,t},
\qquad
Y_{i,t}(1)\gets Y_{i,t}+\delta(1-W_{i,t}).
\]

\BlankLine
\For{$m=1,\ldots,M$}{
\textbf{Step 2 (Generate an SRSB-compatible assignment path).}
\For{$t=1,\ldots,T$}{
\textbf{Construct balancing variables.} Construct $\bH_{i,t}^{(m)}$ in the same way as in the SRSB procedure under $H_0(\delta)$ (e.g., if $\bH_{i,t}$ includes $Y_{i,t-1}$, use $Y_{i,t-1}(W_{i,t-1}^{(m)})$ for $\bH_{i,t}^{(m)}$)\;
\textbf{Rerandomize.} Generate $\bW_{1:N,t}^{(m)}$ by applying the randomization procedure at time $t$ (Algorithm~\ref{alg:srsb}) with balancing variables $\{\bH_{i,t}^{(m)}, 1\leq i \leq N\}$ and threshold $c$\;
}

\BlankLine
\textbf{Step 3 (Compute the test statistic under $\bW_{1:N,1:T}^{(m)}$).}
For each $t\in[T]$, set $Y_{i,t}^{(m)} \gets Y_{i,t}(W_{i,t}^{(m)})$ and compute
\[
\hat{\tau}_t^{(m)}\gets
\frac{2}{N}\sum_{i=1}^N W_{i,t}^{(m)}Y_{i,t}^{(m)}
-\frac{2}{N}\sum_{i=1}^N (1-W_{i,t}^{(m)})Y_{i,t}^{(m)},
\qquad
\hat{\tau}^{(m)}\gets \frac{1}{T}\sum_{t=1}^T \hat{\tau}_t^{(m)}.
\]
}

\BlankLine
\textbf{Step 4 (Monte Carlo $p$-value).}
Compute
\[
p \;\gets\; \frac{1+\sum_{m=1}^M I\!\left(\left|\hat{\tau}^{(m)}\right|\ge \left|\hat{\tau}\right|\right)}{1+M}.
\]
\Return{$p$}\;
\end{algorithm}


\subsubsection{Asymptotic Inference}\label{sec:CLT-nocarryover}
The second approach relies on the large-$T$ behavior of $\hat{\tau}$. In this asymptotic regime, the number of units $N$ may be fixed or may diverge with $T$. We apply a martingale central limit theorem \citep{brown1971martingale, hall2014martingale} to establish asymptotic normality of $\hat{\tau}$ under the SRSB design.

Recall $\{\mathcal{F}_t:t\ge 0\}$ is the filtration generated by the assignment history up to time $t$. By construction, $\hat{\tau}_t$ is $\mathcal{F}_t$-measurable. A key observation is that, at each time $t$, SRSB assigns exactly $N/2$ units to treatment and the acceptance rule is symmetric with respect to swapping treatment and control: if an assignment $\bW_{1:N,t}$ is acceptable, then its complement $\boldsymbol{1}-\bW_{1:N,t}$ is also acceptable because they induce the same Mahalanobis distance for the treated--control difference in balancing variables. This symmetry implies that $\E[W_{i,t}\mid \mathcal{F}_{t-1}]=1/2$, and hence
\[
\E\!\left[\hat{\tau}_t \mid \mathcal{F}_{t-1}\right]=\tau_t,
\]
where $\tau_t$ denotes the finite-population average treatment effect at time $t$. Consequently, the sequence $\{\hat{\tau}_t-\tau_t,\mathcal{F}_t\}_{t\ge 1}$ forms a martingale difference sequence.
Define the predictable variance and total variance as
\[
V_t^2
\;=\;
\sum_{s=1}^t \E\!\left[(\hat{\tau}_s-\tau_s)^2 \mid \mathcal{F}_{s-1}\right],
\qquad
S_t^2
\;=\;
\E\!\left[V_t^2\right].
\]
The asymptotic normality of $\hat{\tau}$ is summarized in the following theorem.

\begin{theorem}\label{thm:CLT-simple}
    Suppose the Lindeberg condition holds: For any $\epsilon >0$,
    \[
    \lim_{T \rightarrow \infty}\frac{1}{S_T^2} \sum_{t=1}^T \E \left[(\hat{\tau}_t-\tau_t)^2 I(|\hat{\tau}_t-\tau_t| \geq \epsilon S_T) \right] = 0, 
    \]
    and $V_T^2/S_T^2 \stackrel{P}{\rightarrow} 1$, then we have
    \[
    \frac{T (\hat{\tau} - \bar{\tau})}{S_T}\stackrel{d}{\rightarrow} N(0,1).
    \]
    A sufficient condition for Lindeberg's condition is that there exists a constant $C>0$ such that $|Y_{i,t}| \leq C$ for all $i \in [N],t \in [T]$ and $S_T^2 \rightarrow \infty$.
\end{theorem}

Although the approximated distribution of the rerandomization estimator $\hat{\tau}_t$ given $\mathcal{F}_{t-1}$ at a fixed time $t$ need not be normal \citep{li2018asymptotic}, averaging over $T$ time periods yields asymptotic normality under the martingale CLT in Theorem~\ref{thm:CLT-simple}. In particular, the conclusion of Theorem~\ref{thm:CLT-simple} holds regardless of whether $N$ is fixed or diverges with $T$, provided the total variance satisfies $S_T^2\to\infty$ and the remaining regularity conditions are met.

The proof of Theorem~\ref{thm:CLT-simple} follows from the martingale CLT \citep{brown1971martingale, hall2014martingale}. A similar proof technique is used in time-series experiments \citep{bojinov2019time, bojinov2021panel}. A key requirement is the stable variance condition $V_T^2/S_T^2 \stackrel{P}{\to} 1$, which stipulates that the realized predictable variance to concentrate around its expectation. Intuitively, this rules out situations where the asymptotic fluctuation is dominated by a few atypical periods with unusually large conditional variance, and it ensures that the random normalization $V_T$ is asymptotically equivalent to the deterministic normalization $S_T$.

To understand the scale of $V_T^2$ and the implied rate for $\hat{\tau}$, we borrow a variance approximation from the rerandomization literature. Let $S_{Y_t(w)}^2$ and $S_{\tau_t}^2$ denote the finite-population variances of $\{Y_{i,t}(w), 1\leq i \leq N\}$ and $\{\tau_{i,t}, 1\leq i \leq N\}$ with $\tau_{i,t}=Y_{i,t}(1)-Y_{i,t}(0)$, respectively. Let $S_{Y_t(w),\bH_t}$ denote the finite-population covariance between $\{Y_{i,t}(w), 1\leq i \leq N\}$ and $\{\bH_{i,t},1\leq i \leq N\}$, and let $S_{\bH_t}^2$ be the finite-population covariance matrix of $\{\bH_{i,t}, 1\leq i \leq N\}$. For large $N$, \citet{li2018asymptotic} showed that under rerandomization the conditional variance admits the approximation
\begin{equation}\label{eq:var-t}
\E\!\left[(\hat{\tau}_t-\tau_t)^2 \mid \mathcal{F}_{t-1}\right]
\;\approx\;
\frac{V_{\tau\tau,t}}{N}\Bigl\{1-(1-v_{d,c})R_t^2\Bigr\},
\end{equation}
where
\[
V_{\tau\tau,t}=2S_{Y_t(1)}^2+2S_{Y_t(0)}^2-S_{\tau_t}^2,
\qquad
v_{d,c}=\frac{\PP(\chi_{d+2}^2\le c)}{\PP(\chi_d^2\le c)},
\]
and
\[
R_t^2
=
\frac{2 S_{Y_t(1)\mid \bH_t}^2+2 S_{Y_t(0)\mid \bH_t}^2-S_{\tau_t\mid \bH_t}^2}{V_{\tau\tau,t}}
\in[0,1].
\]
Here
\[
S_{Y_t(w)\mid \bH_t}^2
=
S_{Y_t(w),\bH_t}\,(S_{\bH_t}^2)^{-1}\,S_{\bH_t,Y_t(w)}
\]
is the finite-population variance of the linear projection of $Y_t(w)$ onto $\bH_t$ (and $S_{\tau_t\mid \bH_t}^2$ is defined analogously).

Equation~\eqref{eq:var-t} suggests that, up to constants, $\E[(\hat{\tau}_t-\tau_t)^2\mid\mathcal{F}_{t-1}]$ is of order $1/N$, and therefore
\[
V_T^2
=
\sum_{t=1}^T \E\!\left[(\hat{\tau}_t-\tau_t)^2\mid \mathcal{F}_{t-1}\right]
\;\asymp\;
\frac{T}{N}.
\]
Since Theorem~\ref{thm:CLT-simple} assumes $V_T^2/S_T^2\stackrel{P}{\to}1$, the requirement $S_T^2\to\infty$ is equivalent to $T/N\to\infty$, i.e.,
\[
\frac{N}{T}\to 0
\qquad\text{as } T\to\infty.
\]
Under this scaling, the standard error of $\hat{\tau}$ is of order $S_T/T\asymp 1/\sqrt{NT}$.

For comparison, under complete randomization at time $t$,
\[
\E\!\left[(\hat{\tau}_t-\tau_t)^2 \mid \mathcal{F}_{t-1}\right]
=
\frac{V_{\tau\tau,t}}{N}.
\]
Thus, relative to complete randomization, rerandomization reduces the period-$t$ variance by the multiplicative factor $1-(1-v_{d,c})R_t^2$ in \eqref{eq:var-t}. The magnitude of the gain is governed by $R_t^2$, which quantifies the extent to which the balancing variables $\{\bH_{i,t}, 1\leq i \leq N\}$ predict the potential outcomes at time $t$. These precision gains then accumulate across time through the averaging in $\hat{\tau}$.

To estimate the asymptotic variance, one can follow \citet{li2018asymptotic} and replace the estimable finite-population quantities $S_{Y_t(w)}^2$, $S_{Y_t(w)\mid \bH_t}^2$, and $S_{\tau_t\mid \bH_t}^2$ in \eqref{eq:var-t} with their sample analogues computed from observed data and the accepted assignment at time $t$. The term $S_{\tau_t}^2$ is not identifiable, but can be conservatively lower-bounded by replacing it with $S_{\tau_t\mid \bH_t}^2$ when forming $V_{\tau\tau,t}$. An alternative is to use a prediction-error-based variance estimator. We detail this approach in Section~\ref{sec:SRSB-carryover}, where a closed-form approximation such as \eqref{eq:var-t} is no longer available.

\section{First-order Carryover effects}\label{sec:SRSB-carryover}

In this section, we extend the SRSB design to settings with carryover effects. We primarily focus on first-order carryover \citep{laird1992analysis, ni2023design, ni2025enhancing} in Assumption \ref{assumption:carryover}(b), while noting that similar ideas can be developed for general $m$-th order carryover structures \citep{bojinov2022}. 


Assumption~\ref{assumption:carryover}(b) rules out longer-lasting effects: treatments prior to time $t-1$ do not affect the outcome at time $t$. Under this assumption, the observed outcome at time $t$ can be written as
\[
Y_{i,t} = Y_{i,t}(W_{i,t-1},W_{i,t}).
\]
In practice, even if carryover effects are higher-order, aggregation of time (i.e., a redefinition of the time scale) can make first-order carryover effects a reasonable approximate assumption. For example, if outcomes depend on treatment exposure over the previous two hours, one may use two-hour periods so that dependence is captured by the immediately preceding period.

Under first-order carryover, we target the following sample average treatment effect (SATE):
\begin{equation}\label{eq:estimand-carryover}
    \bar{\tau} \;=\; \frac{1}{T-1}\sum_{t=2}^{T} \tau_t,
\qquad
\tau_t \;=\; \frac{1}{N}\sum_{i=1}^N \bigl\{Y_{i,t}(1,1) - Y_{i,t}(0,0)\bigr\}.
\end{equation}
In the remainder of this section, we introduce a blocked SRSB design tailored to the first-order carryover setting. We then propose an estimator of $\bar{\tau}$ and establish its theoretical properties.

\subsection{Blocked SRSB Design}

At time $t$, the estimand in \eqref{eq:estimand-carryover} compares average potential outcomes between the two ``stay'' groups
\[
\{i: W_{i,t-1}=W_{i,t}=1\}
\quad\text{and}\quad
\{i: W_{i,t-1}=W_{i,t}=0\}.
\]
To estimate this contrast, it is desirable to choose assignments so that these two groups are comparable, meaning that their pre-treatment characteristics (as captured by the balancing variables) are similar, and each group is approximately representative of the finite population of $N$ units. However, the SRSB procedure from the no-carryover setting enforces balance only between the contemporaneous treated and control groups $\{i:W_{i,t}=1\}$ and $\{i:W_{i,t}=0\}$. Under first-order carryover, this does not guarantee balance between the two ``stay'' groups above, which may differ systematically because they are defined jointly by $(W_{i,t-1},W_{i,t})$. 
To address this issue, we introduce a blocked SRSB design that encourages comparability of the groups $\{i:W_{i,t-1}=W_{i,t}=1\}$ and $\{i:W_{i,t-1}=W_{i,t}=0\}$, as detailed in Algorithm \ref{alg:blocked-SRSB}.

The key idea is to block on the previous assignment at each time $t\ge 2$. Specifically, we partition the $N$ units into two groups,
\[
\mathcal{G}_t^{(1)}=\{i: W_{i,t-1}=1\}
\quad\text{and}\quad
\mathcal{G}_t^{(0)}=\{i: W_{i,t-1}=0\}.
\]
We argue inductively that $\mathcal{G}_t^{(1)}$ and $\mathcal{G}_t^{(0)}$ are comparable and each is approximately representative of the finite population of $N$ units. (At $t=1$, this can be encouraged by balancing baseline covariates, since no outcomes are yet observed.) Conditional on this partition, we apply rerandomization \emph{within} each block to assign $W_{i,t}$. This ensures that, within $\mathcal{G}_t^{(1)}$, the ``stay-treated'' group $\{i: W_{i,t-1}=1,\, W_{i,t}=1\}$ and the corresponding ``switch'' group $\{i: W_{i,t-1}=1,\, W_{i,t}=0\}$ are comparable and each is approximately representative of $\mathcal{G}_t^{(1)}$; an analogous statement holds within $\mathcal{G}_t^{(0)}$. 
Consequently, the four $(W_{i,t-1},W_{i,t})$ groups
\[
\begin{aligned}
    &\,\{i: W_{i,t-1}=1,\, W_{i,t}=1\},\quad
\{i: W_{i,t-1}=1,\, W_{i,t}=0\},\quad\\
&\,\{i: W_{i,t-1}=0,\, W_{i,t}=1\},\quad
\{i: W_{i,t-1}=0,\, W_{i,t}=0\}
\end{aligned}
\]
are each approximately representative of the full population and hence mutually comparable. In particular, this construction implies comparability of the two ``stay'' groups that define our estimand. Moreover, because the next-step blocks are
\[
\mathcal{G}_{t+1}^{(1)}=\{i: W_{i,t}=1\}
\quad\text{and}\quad
\mathcal{G}_{t+1}^{(0)}=\{i: W_{i,t}=0\},
\]
the same representativeness argument carries forward, establishing the inductive hypothesis for time $t+1$.

This blocked design is inspired by stratified rerandomization \citep{wang2023rerandomization}, but here the strata are formed sequentially based on the lagged treatment assignment rather than on baseline discrete covariates. Related ideas have been used in switchback experiments with carryover effects \citep{ni2025enhancing}; we extend this approach by incorporating rerandomization to further promote comparability between the two ``stay'' groups. Blocking also has a practical advantage: it fixes the sizes of the ``stay'' groups at $N/4$ (under equal splits at $t-1$), ensuring that each group contains a stable number of observations. In contrast, under complete randomization at time $t$ or the unblocked SRSB procedure in Section~\ref{sec:SRSB-nocarryover}, the sizes of the ``stay'' groups are random and can be substantially smaller, leading to unstable estimation.

\begin{algorithm}[p] 
\caption{Blocked SRSB Design}\label{alg:blocked-SRSB}
\KwIn{Threshold $c>0$ (e.g., $c=\chi^2_d(\alpha)$) and maximum number of candidate assignments $n$.}
\KwOut{Treatment assignments $\bW_{1:N,1:T}\in\{0,1\}^{N\times T}$ and observed outcomes $\{Y_{i,t}:1\le i\le N,\,1\le t\le T\}$.}

\BlankLine
\For{$t=1,\ldots,T$}{

\textbf{Construct balancing variables:} Compute $\bH_{i,t}\in\mathbb{R}^d$ for each $i\in[N]$\;

\eIf{$t=1$}{
\textbf{Initialization:} Draw $\bW_{1:N,1}$ using either
(i) complete randomization with exactly $N/2$ treated units, or
(ii) rerandomization balancing $\{\bH_{i,1}\}_{i=1}^N$ as in Algorithm~\ref{alg:srsb}\;
}{
\textbf{Form blocks by previous treatment:} Set
\[
\mathcal{G}_t^{(1)} \gets \{i: W_{i,t-1}=1\}, \qquad
\mathcal{G}_t^{(0)} \gets \{i: W_{i,t-1}=0\}.
\]

\textbf{Pre-compute block covariance matrices:} For $g\in\{0,1\}$, compute
\[
\overline{\bH}^{(g)}_t \gets \frac{1}{|\mathcal{G}_t^{(g)}|}\sum_{i\in\mathcal{G}_t^{(g)}} \bH_{i,t},
\qquad
\hat{\bSigma}^{(g)}_t \gets \frac{4}{|\mathcal{G}_t^{(g)}|^2}\sum_{i\in\mathcal{G}_t^{(g)}}
(\bH_{i,t}-\overline{\bH}^{(g)}_t)(\bH_{i,t}-\overline{\bH}^{(g)}_t)^\top .
\]
Initialize $m\gets 0$, $d_{\min}\gets +\infty$, and $\bW_{1:N,t}^{\min}\gets \mathbf{0}$\;

\BlankLine
\While{$m < n$}{
\textbf{Sample within-block candidate assignment:} Initialize $\bW_{1:N,t}^*\gets \mathbf{0}$\;
Randomly select exactly $|\mathcal{G}_t^{(1)}|/2=N/4$ indices from $\mathcal{G}_t^{(1)}$ and set $W^*_{i,t}=1$\;
Randomly select exactly $|\mathcal{G}_t^{(0)}|/2=N/4$ indices from $\mathcal{G}_t^{(0)}$ and set $W^*_{i,t}=1$\;

\textbf{Compute blockwise imbalance vectors:} For $g\in\{0,1\}$, compute
\[
\hat{\btheta}^{(g)*}_t
\;\gets\;
\frac{2}{|\mathcal{G}_t^{(g)}|}\sum_{i\in\mathcal{G}_t^{(g)}} W^*_{i,t}\bH_{i,t}
\;-\;
\frac{2}{|\mathcal{G}_t^{(g)}|}\sum_{i\in\mathcal{G}_t^{(g)}} (1-W^*_{i,t})\bH_{i,t}.
\]
\textbf{Compute blockwise distances:} For $g\in\{0,1\}$, compute
\[
d^{(g)*}_t \;\gets\; (\hat{\btheta}^{(g)*}_t)^\top (\hat{\bSigma}^{(g)}_t)^{-1}\hat{\btheta}^{(g)*}_t.
\]
Set $d^* \gets d^{(1)*}_t + d^{(0)*}_t$\;

\If{$d^* < d_{\min}$}{
$d_{\min}\gets d^*$ and $\bW_{1:N,t}^{\min}\gets \bW_{1:N,t}^*$\;
}

\If{$d^{(1)*}_t < c$ \textbf{and} $d^{(0)*}_t < c$}{
Set $\bW_{1:N,t} \gets \bW_{1:N,t}^*$ and \textbf{break}\;
}

$m \gets m+1$\;
}

\If{$m = n$}{
Set $\bW_{1:N,t} \gets \bW_{1:N,t}^{\min}$\;
}
} 

\BlankLine
\textbf{Implement and observe:} Implement $\bW_{1:N,t}$ and observe $\{Y_{i,t}:1\le i\le N\}$\;
} 

\Return{$\bW_{1:N,1:T}$ and $\{Y_{i,t}:1\le i\le N,\,1\le t\le T\}$}\;
\end{algorithm}

\subsection{Estimation and Inference}
In this section, we study estimation and inference in the presence of first-order carryover effects. For each time $t\ge 2$, a natural estimator of the period-specific effect $\tau_t$ is the difference in means between the two ``stay'' groups:
\[
\hat{\tau}_t
\;=\;
\frac{4}{N}\sum_{i=1}^N W_{i,t-1}W_{i,t} \, Y_{i,t}
\;-\;
\frac{4}{N}\sum_{i=1}^N (1-W_{i,t-1})(1-W_{i,t}) \, Y_{i,t},
\]
where the scaling corresponds to the blocked design in which each stay group has size $N/4$. We then aggregate information across time by averaging these period-specific contrasts,
\begin{equation}\label{eq:est-carryover}
\hat{\tau}
\;=\;
\frac{1}{T-1}\sum_{t=2}^T \hat{\tau}_t .
\end{equation}
Thus, $\hat{\tau}$ compares mean outcomes between the two stay groups at each time point and then averages the resulting contrasts over $t=2,\ldots,T$.

Unlike the no-carryover setting, it is not straightforward to perform randomization inference in this setting. Since there are four potential outcomes, we need three equations about them in the null hypothesis to infer all of them based on the observed outcomes. And it is often difficult to specify and interpret these equations in practice. Hence, we focus on the asymptotic inference in the remaining part of this section.

Similar to the discussion in Section~\ref{sec:CLT-nocarryover}, the SRSB acceptance rule is symmetric in the sense that $\bW_{1:N,t}$ and $\boldsymbol{1}-\bW_{1:N,t}$ are exchangeable given $\mathcal{F}_{t-1}$. Consequently,
\[
\E[W_{i,t}\mid \mathcal{F}_{t-1}] = \frac{1}{2}.
\]
Nevertheless, the period-specific estimator $\hat{\tau}_t$ is generally \emph{not} conditionally unbiased given $\mathcal{F}_{t-1}$. In particular,
\[
\E[\hat{\tau}_t \mid \mathcal{F}_{t-1}]
    \;=\;
      \frac{2}{N}\sum_{i=1}^N W_{i,t-1} \, Y_{i,t}(1,1)
      \;-\; \frac{2}{N}\sum_{i=1}^N (1-W_{i,t-1}) \, Y_{i,t}(0,0)
    \;\neq\; \tau_t,
\]
so the sequence $\{\hat{\tau}_t-\tau_t,\mathcal{F}_t\}_{t\ge 1}$ does \emph{not} form a martingale difference sequence. As a result, the martingale CLT used in Section~\ref{sec:CLT-nocarryover} cannot be applied directly.

However, conditioning one step further back restores unbiasedness:
\begin{equation}\label{eq:lag-one-unbiasedness}
\E[\hat{\tau}_t \mid \mathcal{F}_{t-2}] = \tau_t .
\end{equation}
Because $\hat{\tau}_t$ is not measurable with respect to $\mathcal{F}_{t-1}$, the sequence $\{\hat{\tau}_t-\tau_t,\mathcal{F}_{t-1}\}_{t\ge 2}$ is not a martingale difference sequence, and the martingale CLT arguments used for estimators in panel experiments (e.g., \citealp{bojinov2021panel}) do not apply directly. Instead, \eqref{eq:lag-one-unbiasedness} induces a ``lag-two'' martingale structure that can be analyzed using central limit theorems for mixingales and related dependent arrays \citep{mcleish1977invariance, davidson1992central, de1997central, hall2014martingale}. Let
\[
S_T^2
:=
\Var\!\left(\sum_{t=2}^T (\hat{\tau}_t-\tau_t)\right)
=
\sum_{t=2}^T \E\!\left[(\hat{\tau}_t-\tau_t)^2\right]
+
2\sum_{t=2}^{T-1}\E\!\left[(\hat{\tau}_t-\tau_t)(\hat{\tau}_{t+1}-\tau_{t+1})\right]
\]
denote the total variance of $\sum_{t=2}^T(\hat{\tau}_t-\tau_t)$. We establish the asymptotic distribution of $\hat{\tau}$ in the following theorem.

\begin{theorem}\label{thm:CLT-carryover}
Assume the sequence $\{(\hat{\tau}_t - \tau_t)^2 / \|\hat{\tau}_t - \tau_t\|^2 : 2 \le t \le T\}$
is uniformly integrable, and that
\begin{equation}\label{eq:bd-var}
    \frac{T \max_{2 \le t \le T} \Var(\hat{\tau}_t)}{S_T^2} = O(1).
\end{equation}
Let $b_T \in \mathbb{N}_+$ be a sequence such that $b_T \to \infty$ and $T/b_T \to \infty$, and set
\[
t_j = (j-1)b_T + 1, \, 1 \leq j \leq J=\lceil T/b_T \rceil, \, t_{J+1}=T+1,
\]
where $\lceil x \rceil$ is the smallest integer that is greater than or equal to $x$. Define the blockwise predictable variance
    \[
    V_J^2 =  \sum_{j=1}^J \E \left [ \left(\sum_{t=t_j+1}^{t_{j+1}-1} (\hat{\tau}_t -\tau_t) \right)^2 \middle | \mathcal{F}_{t_j-1} \right],
    \]
and assume the predictable variance of the block-level martingale differences satisfies
\[
V_J^2 / \E[V_J^2]
\stackrel{P}{\longrightarrow} 1
\qquad \text{as } T \to \infty.
\]
Then
\[
\frac{(T-1) (\hat{\tau}-\bar{\tau})}{S_T}
\;\xrightarrow{d}\;
N(0,1).
\]
\end{theorem}

To establish the asymptotic normality, here we follow the ``Bernstein sums'' argument used in \citet{de1997central} and \citet{davidson1992central}. 
Partition the $T$ periods into $J$ blocks with block size $b_T$ (the last block may have a different block size), where block $j$ starts at $t_j$ and ends at $t_{j+1}-1$. 
We can then decompose the normalized sum as
\begin{equation}\label{eq:carryover-decompose}
    \begin{aligned}
    \frac{\sum_{t=2}^T (\hat{\tau}_t - \tau_t)}{S_T}
    &= \frac{1}{S_T} \Bigg\{ \sum_{t=t_1+1}^{t_{2}-1} (\hat{\tau}_t - \tau_t)+
        \sum_{j=2}^{J} \left[
            \hat{\tau}_{t_j} - \tau_{t_j}
            + \sum_{t=t_j+1}^{t_{j+1}-1} (\hat{\tau}_t - \tau_t)
        \right]
    \Bigg\}.
    \end{aligned}
\end{equation}
Denote $Z_j = \sum_{t=t_j+1}^{t_{j+1}-1} (\hat{\tau}_t - \tau_t)/S_T$.
The dominant component of the sum is given by
\[
\sum_{j=1}^{J} Z_j 
\quad=\quad
\frac{1}{S_T} \sum_{j=1}^{J} \sum_{t=t_j+1}^{t_{j+1}-1} (\hat{\tau}_t - \tau_t).
\]
Since we drop the first observation $t_j$ in each block, the block‐level term $Z_j$ satisfies 
\[
\E[Z_j \mid \mathcal{F}_{t_j-1}] = 0,
\]
so that $\{Z_j,\,\mathcal{F}_{t_{j+1}-1},\,1 \le j \le J\}$ forms a martingale difference sequence.
Uniform integrability ensures that no individual term dominates the sum, allowing Lindeberg’s condition for the array $\{Z_j, 1\leq j \leq J\}$ to hold. Combined with Assumption~2, the martingale central limit theorem then yields asymptotic normality of the dominant term $\sum_{j=1}^J Z_j$. Finally,
equation~\eqref{eq:bd-var} guarantees that the remaining terms in the decomposition are asymptotically negligible. This condition is satisfied, for example, when 
$\max_{2 \le t \le T} \Var(\hat{\tau}_t) = O(1/N)$ 
and 
$S_T^2 = \Omega(T/N)$. 



It is worth noting that our theoretical results apply more broadly to any adaptive design that satisfies the lag-two unbiasedness condition $\E[\hat{\tau}_t \mid \mathcal{F}_{t-2}] = \tau_t$, which includes both the unblocked SRSB design in Section~\ref{sec:SRSB-nocarryover} and the blocked SRSB design in Section~\ref{sec:SRSB-carryover}.

To obtain a feasible and conservative estimator of the variance of
\[
\hat{\tau} = \frac{1}{T-1} \sum_{t=2}^T (\hat{\tau}_t - \tau_t),
\]
we build on the martingale–block argument used in our asymptotic theory.  
Recall that the limiting variance is driven by the predictable variance
\[
V_J^2
= \sum_{j=1}^J 
\mathbb{E}\!\left[
    \left( \sum_{t=t_j+1}^{t_{j+1}-1} (\hat{\tau}_t - \tau_t) \right)^2 
    \,\middle|\, \mathcal{F}_{t_j-1}
\right].
\]
A practical approach is to replace each block's conditional variance contribution by a 
\emph{prediction-based residual}, relying only on quantities observable up to time $t_j - 1$.  
For each block $j$, let
\[
M_j = M_j(\mathcal{F}_{t_j-1})
\]
be a predictor of the future sum $\sum_{t=t_j+1}^{t_{j+1}-1} \hat{\tau}_t$, where $M_j$ is measurable with respect to $\mathcal{F}_{t_j-1}$.  
Natural choices include, for example, a scaled average of past estimates,
\[
M_j 
= \frac{t_{j+1} - t_j - 1}{t_j - 1} 
  \sum_{t=1}^{t_j - 1} \hat{\tau}_t,
\]
or a weighted version thereof.  
We then estimate the $j$-th block's variance contribution using the residual
\[
\widehat{V}_{J,j}^2
    = \left( \sum_{t=t_j+1}^{t_{j+1}-1} \hat{\tau}_t - M_j \right)^2.
\]
Since 
\[
\mathbb{E}\!\left[\sum_{t=t_j+1}^{t_{j+1}-1} \hat{\tau}_t \,\middle|\, \mathcal{F}_{t_j-1}\right]
    = \sum_{t=t_j+1}^{t_{j+1}-1} \tau_t,
\]
we have
\begin{equation}\label{eq:conservative-var-carryover}
\mathbb{E}\!\left[\widehat{V}_{J,j}^2 \,\middle|\, \mathcal{F}_{t_j-1}\right]
\;\ge\;
    \mathbb{E}\!\left[
        \left( \sum_{t=t_j+1}^{t_{j+1}-1} (\hat{\tau}_t - \tau_t) \right)^2
        \,\middle|\, \mathcal{F}_{t_j-1}
    \right].
\end{equation}
Thus $\widehat{V}_{J,j}^2$ is conservative for the true block-level variance term in the sense of \eqref{eq:conservative-var-carryover}.  
Aggregating across blocks yields the prediction-based conservative estimator
\begin{equation}\label{eq:var-est}
   \widehat{V}_J^2 = \sum_{j=1}^J \widehat{V}_{J,j}^2,
\end{equation}
and we estimate the variance of $\hat{\tau}$ by $\widehat{V}_J^2 / (T-1)^2$.
This estimator is conservative whenever $M_j$ is an imperfect yet measurable predictor of the conditional mean, and it preserves the martingale structure underlying the asymptotic theory.  
It requires no access to unobserved potential outcomes and performs well in settings where treatment effects evolve smoothly over time. The variance of the estimator \eqref{eq:est-nocarryover} in Section \ref{sec:SRSB-nocarryover} can also be similarly estimated. Additional simulation studies illustrating the properties of Wald-style confidence intervals based on this conservative variance estimator are provided in the supplementary material.



\section{Simulation Studies}\label{sec:simulation}

\subsection{No Carryover Effects}\label{sec:simu-simple-nocarryover}
Consider the following auto-regressive Data-Generating Process (DGP): Let $U_t$ be an AR(1) model with covariate
\[
U_{i,t} = 0.8U_{i,t-1}+\rho X_{i,t}+\epsilon_{i,t}, \, X_{i,t} \sim N(0,1), \,\epsilon \sim N(0,0.25).
\]
In our first no-carryover setting, set $Y_{i,t}(0)= U_{i,t}$ and $Y_{i,t}(1)= U_{i,t}+0.5$ so that the treatment effect is $0.5$.
Since we work in the finite-population framework, we begin each simulation by generating and fixing the potential outcomes and covariates $\{Y_{i,t}(w), X_{i,t} : 1 \leq i \leq N,\ 1 \leq t \leq T,\ w \in \{0,1\}\}$. We then conduct $M = 500$ experimental replications, each time estimating the treatment effect and computing the Root Mean Squared Error (RMSE) of the resulting estimator.

The baseline strategy is a completely randomized experiment: for each time period $t$, we randomly assign $N/2$ units to treatment and $N/2$ to control. We compare this with the proposed strategy using sequential rerandomization, which balances the previous outcomes $\{Y_{i,t-1}, 1\leq i \leq N\}$ and covariate $\{X_{i,t}, 1\leq i \leq N\}$ to exploit temporal auto-correlation of the potential outcomes and covariate information. The threshold $c$ for accepting a treatment assignment is set to $\chi_2^2(0.01)$, yielding an acceptance probability of approximately 0.01 for each candidate assignment.

In both designs, we estimate the sample average treatment effect using the standard difference-in-means estimator $\hat{\tau}$. We consider three simulation settings:

\begin{itemize}
    \item Fix $T = 20, \rho = 1$ and vary $N \in \{100, 200, \dots, 1000\}$.
    \item Fix $N = 200, \rho=1$ and vary $T \in \{10,20,\dots,100\}$.
    \item Fix $T=20, N=200$ and vary $\rho \in \{0.5, 0.6, \dots, 1.5\}$.
\end{itemize}
The resulting RMSEs on the log scale for the first two settings are shown in Figure \ref{fig:simu-simple1}. We observe that balancing the previous outcomes via the sequential rerandomization procedure consistently reduces the estimation error of the treatment effect, compared with the completely randomized experiments. The curves on the log scale are close to lines with slope $-0.5$ and illustrates that the estimation rates are parametric. 

\begin{figure}[t]
	\centering
	\subfigure[$T=20, \rho=1, N \in \{100, \dots, 1000\}$]{
		\begin{minipage}[t]{0.48\linewidth}
			\centering
			\includegraphics[width=2.5in]{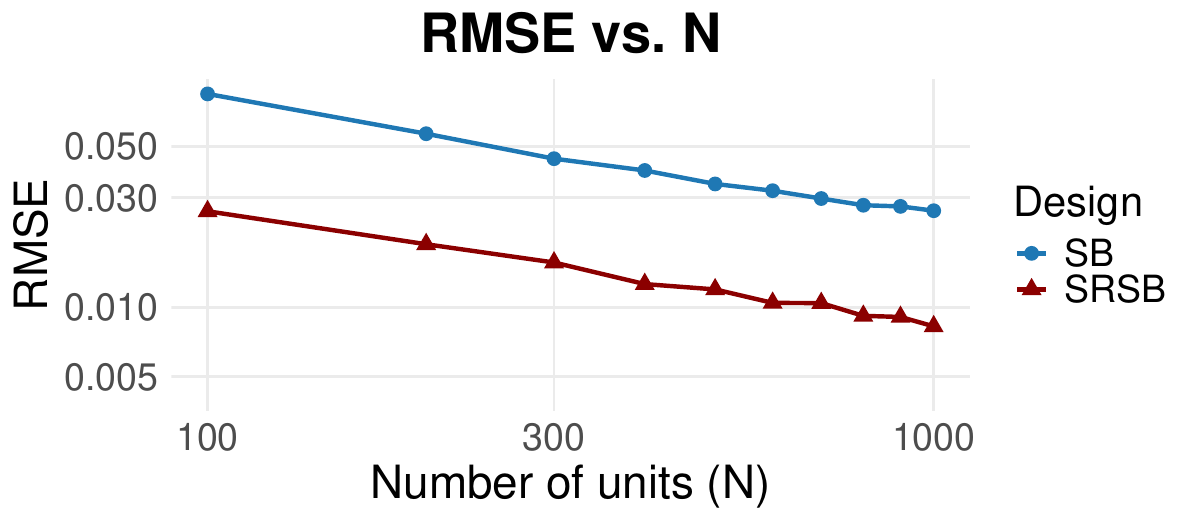}
	\end{minipage}}
	\subfigure[$N=200, \rho=1, T \in \{10,20,\dots,100\}$]{
		\begin{minipage}[t]{0.48\linewidth}
			\centering
			\includegraphics[width=2.5in]{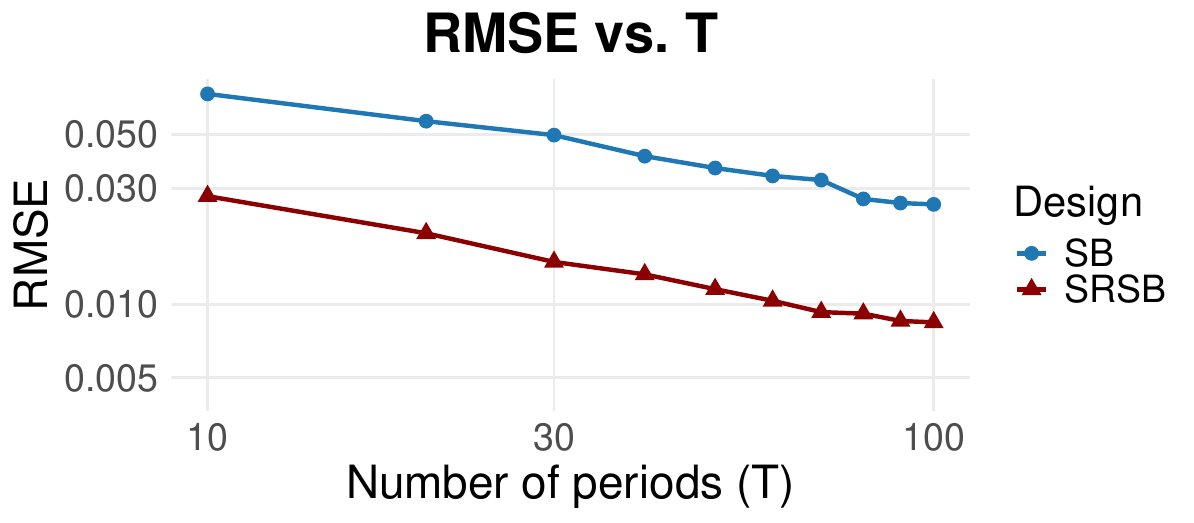}
	\end{minipage}}
	\centering
	\caption{RMSE comparison between completely randomized experiments and sequential rerandomization experiments in the no-carryover setting. The potential outcomes are generated from an AR(1) model with covariates; see Section~\ref{sec:simu-simple-nocarryover} for details of the data-generating process.}
	\label{fig:simu-simple1}
\end{figure}

\begin{figure}[t]
	\centering
	\subfigure[$T=20, N=200, \rho \in \{0.5,\dots,1.5\}$]{
		\begin{minipage}[t]{0.48\linewidth}
			\centering
			\includegraphics[width=2.5in]{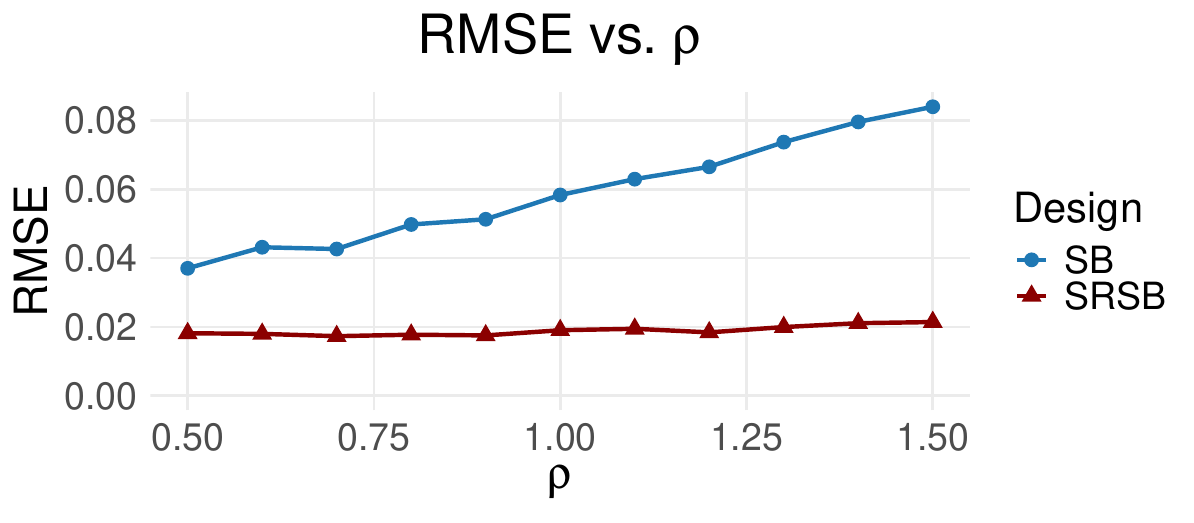}
	\end{minipage}}
	\subfigure[Proportion of variance reduction vs $\rho$]{
		\begin{minipage}[t]{0.48\linewidth}
			\centering
			\includegraphics[width=2.5in]{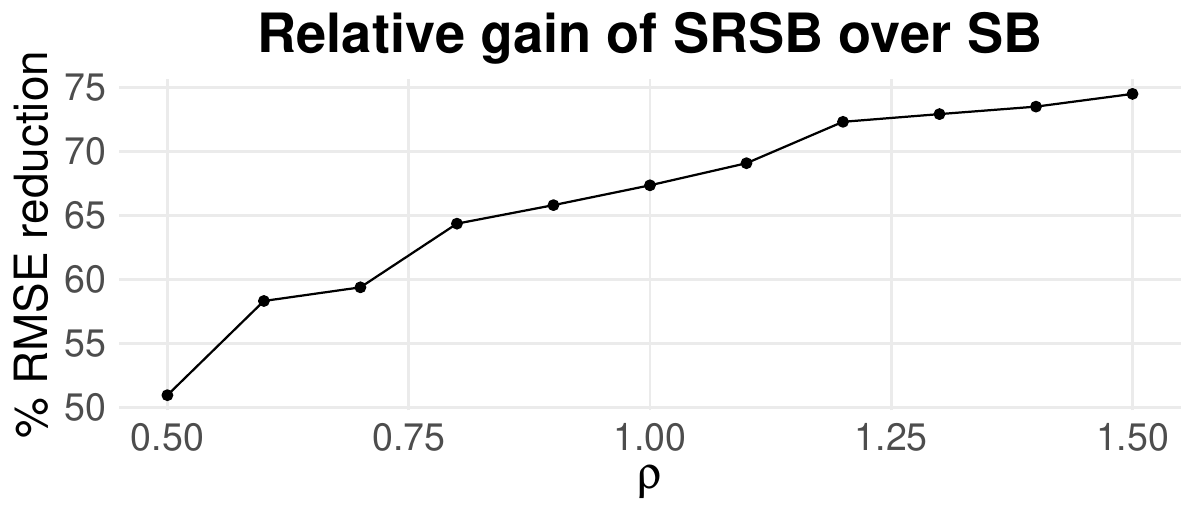}
	\end{minipage}}
	\centering
	\caption{RMSE vs $\rho$ and proportion of RMSE reduced by sequential rerandomization. The potential outcomes are generated from an AR(1) model with covariates; see Section~\ref{sec:simu-simple-nocarryover} for details of the data-generating process.}
	\label{fig:simu-simple2}
\end{figure}
Figure~\ref{fig:simu-simple2} reports the RMSEs and the proportion of RMSE reduction achieved by sequential rerandomization in the third setting. In Figure~\ref{fig:simu-simple2}(a), the RMSE under both designs increases as~$\rho$ grows. This is expected: as~$\rho$ becomes larger, the covariate becomes more prognostic, and the variance component of the difference-in-means estimator attributable to covariate imbalance increases. Nevertheless, the estimator under sequential rerandomization consistently attains a smaller RMSE than under complete randomization.
Figure~\ref{fig:simu-simple2}(b) further shows that the proportion of RMSE reduction, defined as $1-\text{RMSE}_{\text{SRSB}}/\text{RMSE}_{\text{CR}}$, increases with~$\rho$. As the variance contribution from the covariate becomes larger, balancing it yields greater gains in variance reduction, leading to a larger relative improvement of sequential rerandomization over complete randomization.

\subsection{Carryover Effects}\label{sec:simu-carryover}

Consider the following auto-regressive Data-Generating Process (DGP):
\[
U_{i,t} = 0.7 U_{i,t-1}+X_{i,t}+\epsilon_{i,t},
\]
\[
X_{i,t} \sim N(0,1), \, \epsilon \sim N(0,0.25).
\]
In our setting with first-order carryover effects, we let 
\[
Y_{i,t}(0,0)=U_{i,t}, \, Y_{i,t}(0,1)=U_{i,t}+1,
\]
\[
Y_{i,t}(1,0)=U_{i,t}+0.5, \, Y_{i,t}(1,1)=U_{i,t}+3.5.
\]
Similar to the previous section, we first generate the potential outcomes and covariates $\{Y_{i,t}(w_1,w_2), X_{i,t} : 1 \leq i \leq N,\ 1 \leq t \leq T,\ w_1,w_2 \in \{0,1\}\}$. We then conduct $M = 500$ experimental replications. In each replication, we estimate the treatment effect under different experimental designs and compute the RMSE of the resulting estimator $\hat{\tau}$ in \eqref{eq:est-carryover}, defined as the difference in average outcomes between the two ``stay'' groups.

The baseline designs are the completely randomized experiment and its blocked version. 
We also implement the sequential rerandomization design without blocking, introduced in Section~\ref{sec:SRSB-nocarryover}, which balances the observed previous outcomes $\{Y_{i,t-1} : 1 \le i \le N\}$ and the covariates $\{X_{i,t} : 1 \le i \le N\}$. 
We then compare these two designs with the blocked sequential rerandomization design described in Section~\ref{sec:SRSB-carryover}. 
For both sequential rerandomization designs, we set the threshold $c = \chi_2^2(0.01)$.

\begin{figure}[t]
	\centering
	\subfigure[$T=40, N \in \{200, 300, \dots, 1000\}$]{
		\begin{minipage}[t]{0.48\linewidth}
			\centering
			\includegraphics[width=2.5in]{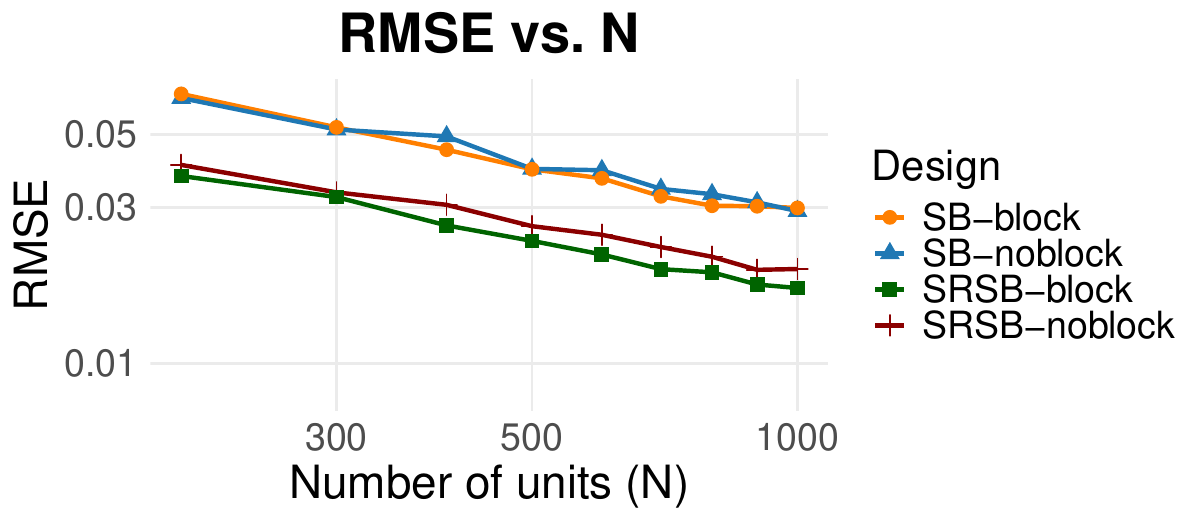}
	\end{minipage}}
	\subfigure[$N=600, T \in \{10,20,\dots,100\}$]{
		\begin{minipage}[t]{0.48\linewidth}
			\centering
			\includegraphics[width=2.5in]{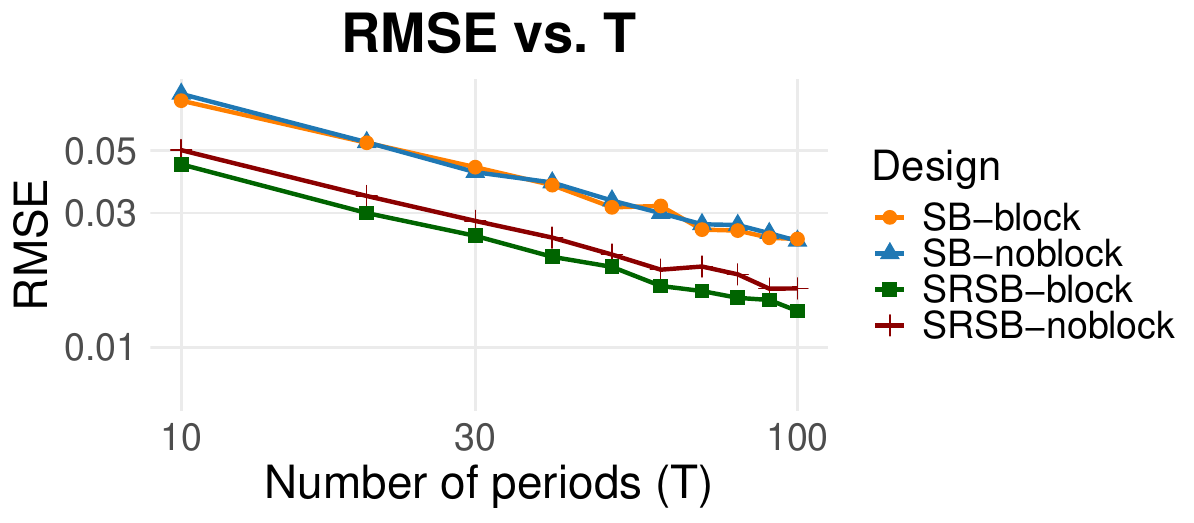}
	\end{minipage}}
	\centering
	\caption{Comparison of RMSE under different experimental designs in a setting with first-order carryover effects. The potential outcomes are generated from an AR(1) model with covariates; see Section~\ref{sec:simu-carryover} for details of the data-generating process. }
	\label{fig:simu-carryover1}
\end{figure}

In both designs, we estimate the sample average treatment effect using the standard difference-in-means estimator $\hat{\tau}$. 
We consider two simulation settings:

\begin{itemize}
    \item Fix $T = 40$ and vary $N \in \{200, 300, \dots, 1000\}$.
    \item Fix $N = 600$ and vary $T \in \{10,20,\dots,100\}$.
\end{itemize}

The resulting RMSEs on the log scale for the first two settings are shown in Figure~\ref{fig:simu-carryover1}. 
We find that the blocked sequential rerandomization procedure consistently yields smaller estimation error than both the completely randomized experiment and the sequential rerandomization without blocking. 
As in the previous experiment, the RMSE curves on the log scale are close to straight lines with slope $-0.5$, illustrating the parametric rates in this setting.

\subsection{Application to Penn World Table}\label{sec:simu-gdp}

In this section, we conduct simulations based on the Penn World Table \citep{feenstra2015next}. The dataset contains annual real GDP for $N = 111$ countries over $T = 48$ consecutive years (1959--2007). Following \citet{arkhangelsky2021synthetic}, we first fit a rank-two latent factor model to the data and obtain the latent factors for unit $i$ at time $t$ as $L_{i,t}$. We then simulate the baseline outcomes by $Y_{i,t}^{\text{base}}=L_{i,t} + e_{i,t}$, where for each $i$, ${e_{i,t}, 1 \leq t \leq T}$ is sampled from an AR(2) model fitted using the residual of the latent factor model. This data-generating process then preserves the dependence structure observed in real macroeconomic panels. We then run a series of simulations using the semi-synthetic baseline outcomes ${Y_{i,t}^{\text{base}}}$.

\subsubsection{No Carryover Effects}\label{sec:gdp-nocarryover}
In the setting without carryover effects, we let
\[
Y_{i,t}(0) = Y_{i,t}^{\text{base}},\, Y_{i,t}(1) = Y_{i,t}^{\text{base}}+\tau
\]
so that the sample ATE satisfies $\bar{\tau}=\tau$. We vary the treatment effect size $\tau \in \{0, 0.1, \dots, 1\}$ to examine how signal strength influences the performance of sequential rerandomization. We compare the SRSB design in Section~\ref{sec:SRSB-nocarryover}, which balances the previous outcomes at each time $t$, with a completely randomized experiment. The blocked design introduced in Section~\ref{sec:SRSB-carryover} is also included for reference. The acceptance threshold is set to $c = \chi_1^2(0.01)$. The resulting RMSEs of the difference-in-means estimator are reported in Figure~\ref{fig:simu-gdp}(a). Consistent with our earlier simulations, sequential rerandomization yields uniformly smaller estimation error than complete randomization. Because there are no carryover effects, blocking on the previous assignments \(\bW_{1:N,t-1}\) provides no additional benefit once we have already balanced on the lagged outcomes \(\{Y_{i,t-1}: 1\le i\le N\}\).

As the treatment effect ${\tau}$ increases, the RMSE under complete randomization remains essentially unchanged, indicating that its performance is insensitive to the effect size. 
In contrast, the RMSE under sequential rerandomization increases with ${\tau}$, reducing its relative improvement. 
This pattern arises because the benefit of rerandomization depends on how predictive the previously observed outcomes are for the current potential outcomes. 
When ${\tau} = 0$, the observed lagged outcomes coincide with the lagged potential outcomes; since these are generated from a latent factor model with serial dependence, balancing previous potential outcomes is highly informative. 
However, when ${\tau}$ grows, the observed lagged outcomes become a mixture of treated and untreated potential outcomes, making them less predictive of the current potential outcomes. 
Consequently, the efficiency gain from SRSB diminishes as the treatment effect increases.

\begin{figure}[t]
	\centering
	\subfigure[No carryover]{
		\begin{minipage}[t]{0.48\linewidth}
			\centering
			\includegraphics[width=2.5in]{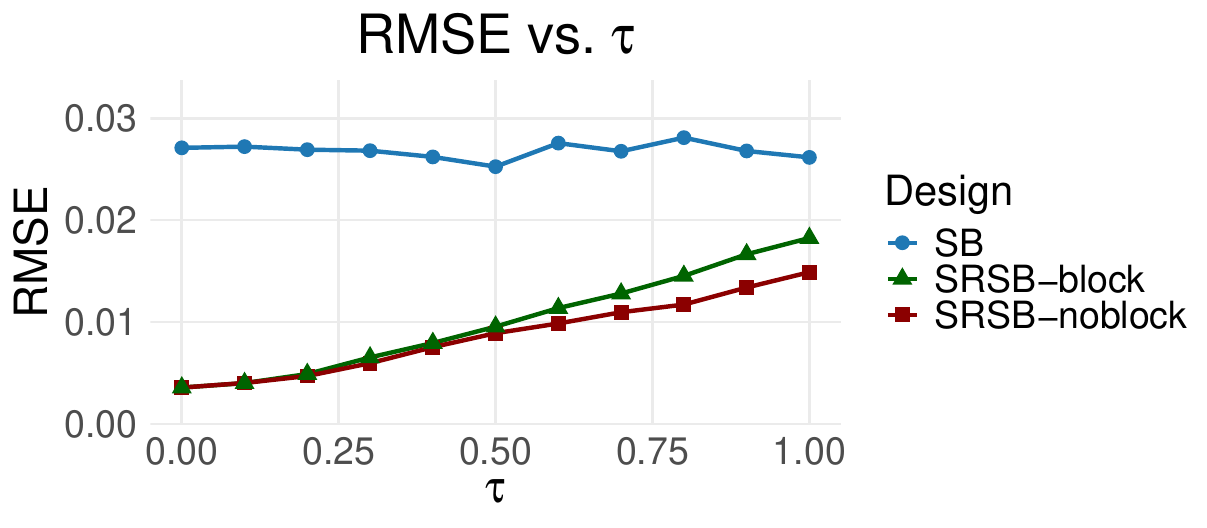}
	\end{minipage}}
	\subfigure[First-order carryover]{
		\begin{minipage}[t]{0.48\linewidth}
			\centering
			\includegraphics[width=2.5in]{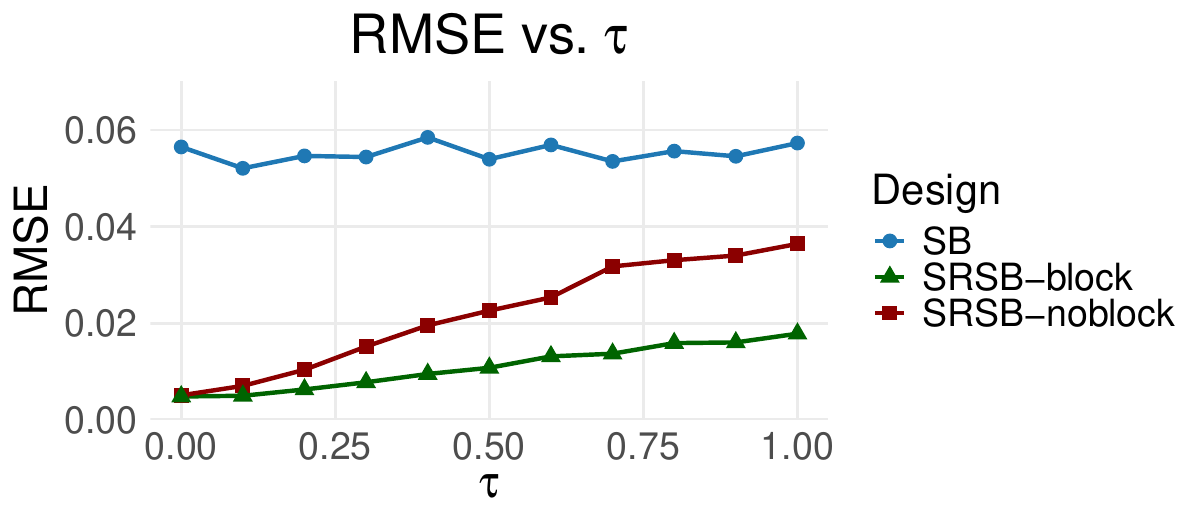}
	\end{minipage}}
	\centering
	\caption{Comparison of RMSE across different treatment effect sizes. The potential outcomes are generated from a factor model fitted to the Penn World Table; see Section~\ref{sec:simu-gdp} for details.}
	\label{fig:simu-gdp}
\end{figure}

\subsubsection{First-order Carryover Effects}\label{sec:gdp-carryover}
In the setting with first-order carryover effects, we let
\[
Y_{i,t}(0,0) = Y_{i,t}^{\text{base}},\,Y_{i,t}(0,1) = Y_{i,t}^{\text{base}}+\tau,
\]
\[
Y_{i,t}(1,0) = Y_{i,t}^{\text{base}}, \,Y_{i,t}(1,1) = Y_{i,t}^{\text{base}}+2\tau,
\]
so that the sample ATE satisfies $\bar{\tau}=2\tau$. We again vary $\tau \in \{0, 0.1, \dots, 1\}$ to examine how signal strength influences the performance of sequential rerandomization. We compare three designs: the sequential rerandomization procedure in Section~\ref{sec:SRSB-nocarryover} and the blocked SRSB design in Section~\ref{sec:SRSB-carryover} (both of which balance the previous outcomes at each time $t$), and a completely randomized experiment. The acceptance threshold is set to $c = \chi_1^2(0.01)$. The resulting RMSEs of the estimator in \eqref{eq:est-carryover} are reported in Figure~\ref{fig:simu-gdp}(b). The blocked SRSB design yields uniformly smaller estimation error than the other two approaches.

Similarly, as \(\tau\) increases, the RMSEs of the two sequential rerandomization procedures increase for the same reason discussed above. Notably, the RMSE of the blocked SRSB design grows more slowly than that of the unblocked design, suggesting that blocking makes the procedure less sensitive to effect size and therefore more robust in the presence of first-order carryover effects. As emphasized in the design’s motivation, blocking on \(W_{i,t-1}\) and then applying rerandomization separately within each block yields two comparable ``stay'' groups of fixed size, which stabilizes estimation and improves precision when there exist first-order carryover effects.

\subsubsection{A Markovian Carryover Model with Latent States}\label{sec:gdp-markov}
In this section, we augment the baseline outcomes with MDP-style carryover effects to better mimic realistic applications. We first draw the state and outcome shocks $\{\nu_{i,t}: 1 \le i \le N,\ 1 \le t \le T\}$ and $\{\xi_{i,t}: 1 \le i \le N,\ 1 \le t \le T\}$ i.i.d.\ from $N(0,0.16)$. For each unit $i$, we define a latent state process recursively by
\[
S_{i,1}=0, \,S_{i,t} = \rho\, S_{i,t-1} + \, W_{i,t-1} + \nu_{i,t},
\quad t=2,\ldots,T,
\]
and generate outcomes according to
\[
Y_{i,t}
=
Y^{\mathrm{base}}_{i,t}
\;+\;
0.5\, \text{tanh}\!\left(S_{i,t}\right)
\;+\;
0.5\, W_{i,t}\, \text{tanh}\!\left(S_{i,t}\right)
\;+\;
\xi_{i,t},
\quad t=1,\ldots,T.
\]
Under the design-based (finite-population) perspective, the collections ${Y^{\mathrm{base}}_{i,t}}$, ${\nu_{i,t}}$, and ${\xi_{i,t}}$ are treated as fixed, so that all randomness arises solely from the treatment assignment.

In this setting, the parameter $\rho$ determines the persistence of the latent state and therefore the strength of carryover effects. When $\rho=0$, the state has no persistence and the carryover is effectively first-order; in this case, our estimator \eqref{eq:est-carryover} is unbiased. When $\rho\neq 0$, the outcome at time $t$ can depend on the entire treatment history up to time $t$ through the state recursion, and the influence of treatments prior to $t-1$ becomes more pronounced as $|\rho|$ increases.
Figure~\ref{fig:simu-gdp-mdp-traj1}--\ref{fig:simu-gdp-mdp-traj2} illustrate this mechanism using the outcome trajectories of two randomly selected units under different values of $\rho$. We plot the realized outcomes under the all-ones path, the all-zeros path, and the alternating-block path
\[
\bW_{i,1:T}=\bigl(\boldsymbol{0}_4^{\top},\, \boldsymbol{1}_4^{\top},\, \boldsymbol{0}_4^{\top},\, \boldsymbol{1}_4^{\top},\, \dots \bigr)\in \{0,1\}^{48},
\]
where the final assignment corresponds to a time-level block schedule, as in \citet{hu2024switchbackexperimentsgeometricmixing}. Their notion of blocking refers to grouping consecutive time periods into treatment blocks and should not be confused with the blocked SRSB design in Section~\ref{sec:SRSB-carryover}. The carryover is visibly stronger when $\rho=0.5$: after a switch, the realized trajectory approaches the corresponding steady ``all-treated'' or ``all-control'' path more slowly than when $\rho=0.1$.

\begin{figure}[t]
	\centering
	\subfigure{
		\begin{minipage}[t]{0.48\linewidth}
			\centering
			\includegraphics[width=3in]{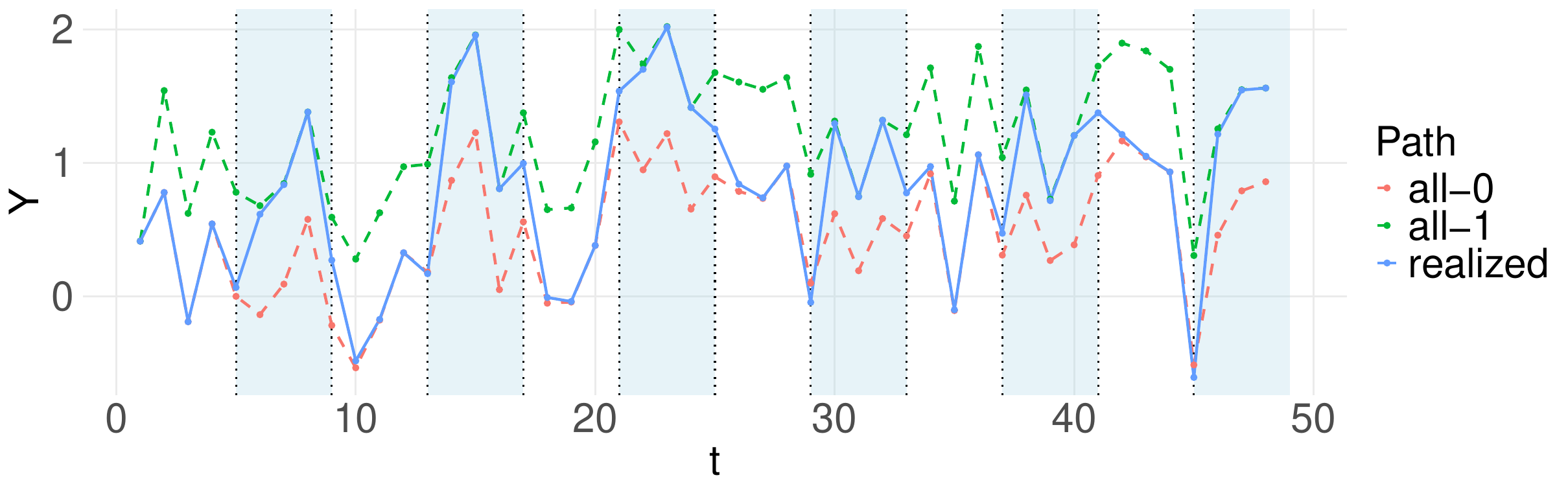}
	\end{minipage}}
    \subfigure{
		\begin{minipage}[t]{0.48\linewidth}
			\centering
			\includegraphics[width=3in]{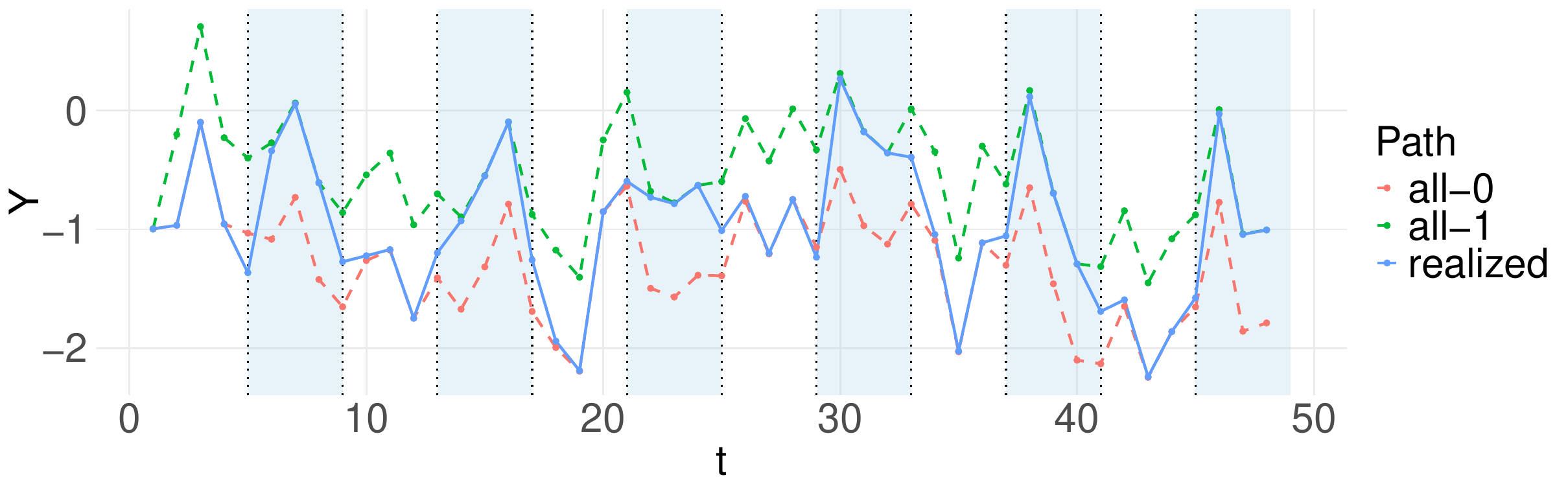}
	\end{minipage}}
	\centering
    \caption{Simulated trajectories for two randomly selected units under a Markovian carryover model with latent states and $\rho = 0.1$, based on the Penn World Table. See Section~\ref{sec:gdp-markov} for further details.}
    \label{fig:simu-gdp-mdp-traj1}
\end{figure}

\begin{figure}[t]
	\centering
	\subfigure{
		\begin{minipage}[t]{0.48\linewidth}
			\centering
			\includegraphics[width=3in]{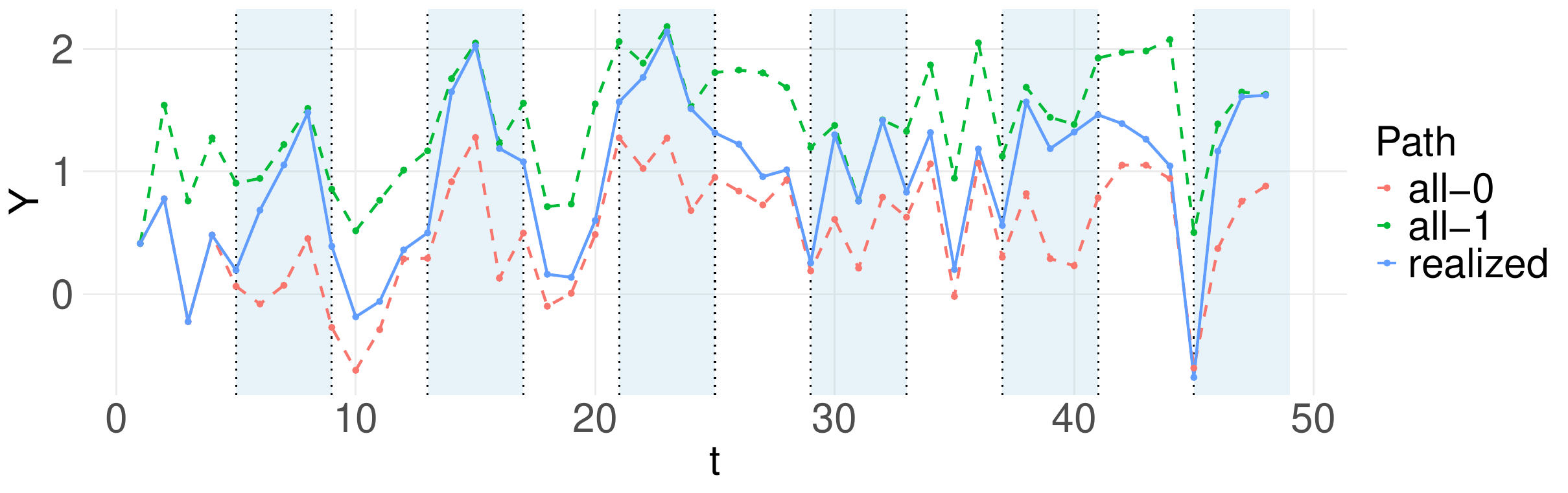}
	\end{minipage}}
    \subfigure{
		\begin{minipage}[t]{0.48\linewidth}
			\centering
			\includegraphics[width=3in]{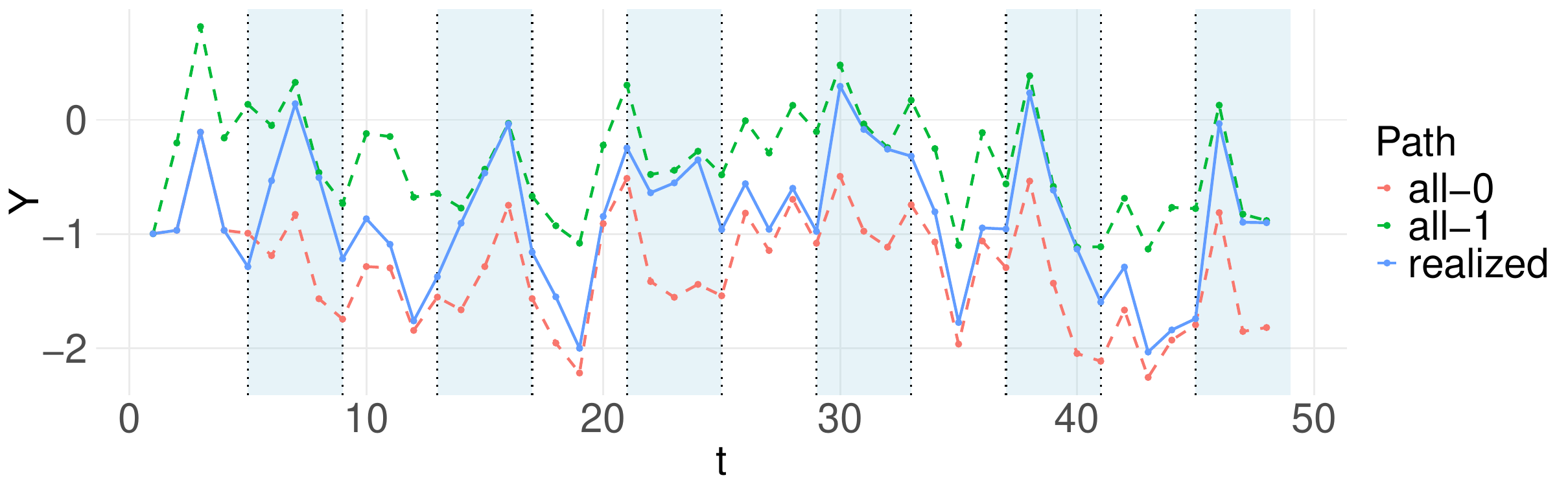}
	\end{minipage}}
	\centering
    \caption{Simulated trajectories for two randomly selected units under a Markovian carryover model with latent states and $\rho = 0.5$, based on the Penn World Table. See Section~\ref{sec:gdp-markov} for further details.}
    \label{fig:simu-gdp-mdp-traj2}
\end{figure}

In this simulation, we examine how the performance of the estimator in \eqref{eq:est-carryover} varies with the state persistence parameter $\rho$. The target estimand is the total effect of the all-ones versus all-zeros treatment paths,
\[
\bar{\tau}
=
\frac{1}{N(T-1)}
\sum_{i=1}^N
\sum_{t=2}^T
\Bigl\{
Y_{i,t}(\boldsymbol{1}_t)-Y_{i,t}(\boldsymbol{0}_t)
\Bigr\}.
\]
Recall that the estimator \eqref{eq:est-carryover} relies on a first-order carryover approximation and is therefore biased when $\rho\neq 0$. Figure~\ref{fig:simu-gdp-mdp} reports the resulting (absolute) bias, variance, and RMSE as functions of $\rho$. As expected, the bias increases with $\rho$: larger $\rho$ induces more persistent carryovers, making the first-order approximation less accurate. Both SRSB designs yield smaller variances than complete randomization, whereas the additional variance reduction from the blocked SRSB design is not pronounced in this setting. When $\rho$ is small so that bias is negligible relative to variance, SRSB achieves smaller RMSE than complete randomization. As $\rho$ increases, however, the bias begins to dominate the error, and the overall advantage of rerandomization becomes less evident.

\begin{figure}[t]
	\centering
	\subfigure[Bias]{
		\begin{minipage}[t]{0.48\linewidth}
			\centering
			\includegraphics[width=2.5in]{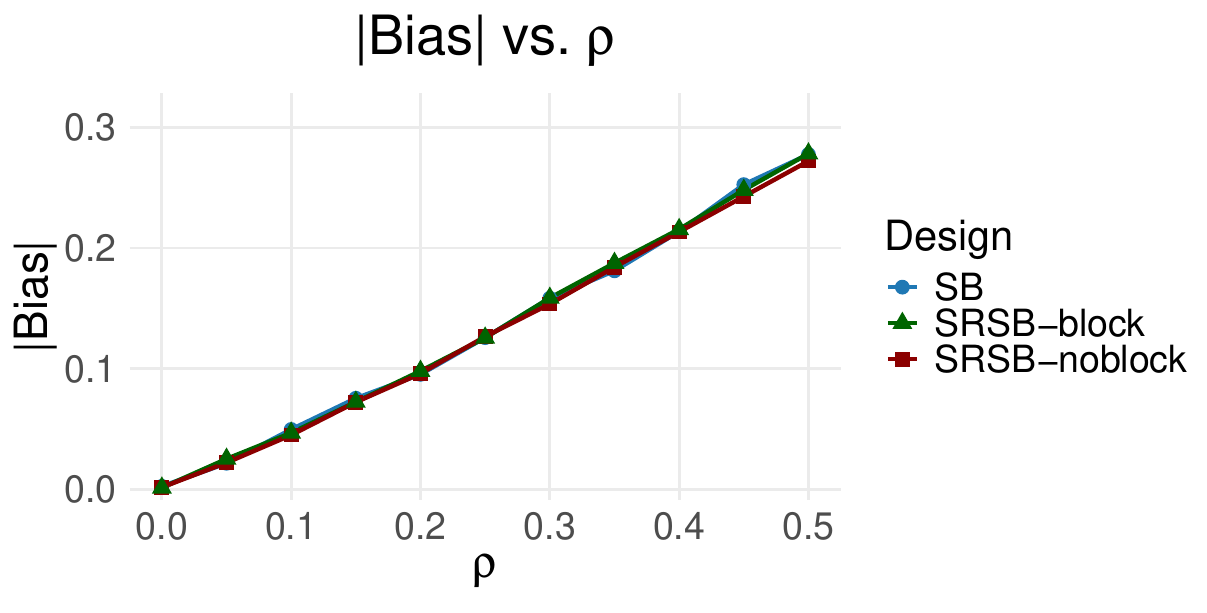}
	\end{minipage}}
	\subfigure[Variance]{
		\begin{minipage}[t]{0.48\linewidth}
			\centering
			\includegraphics[width=2.5in]{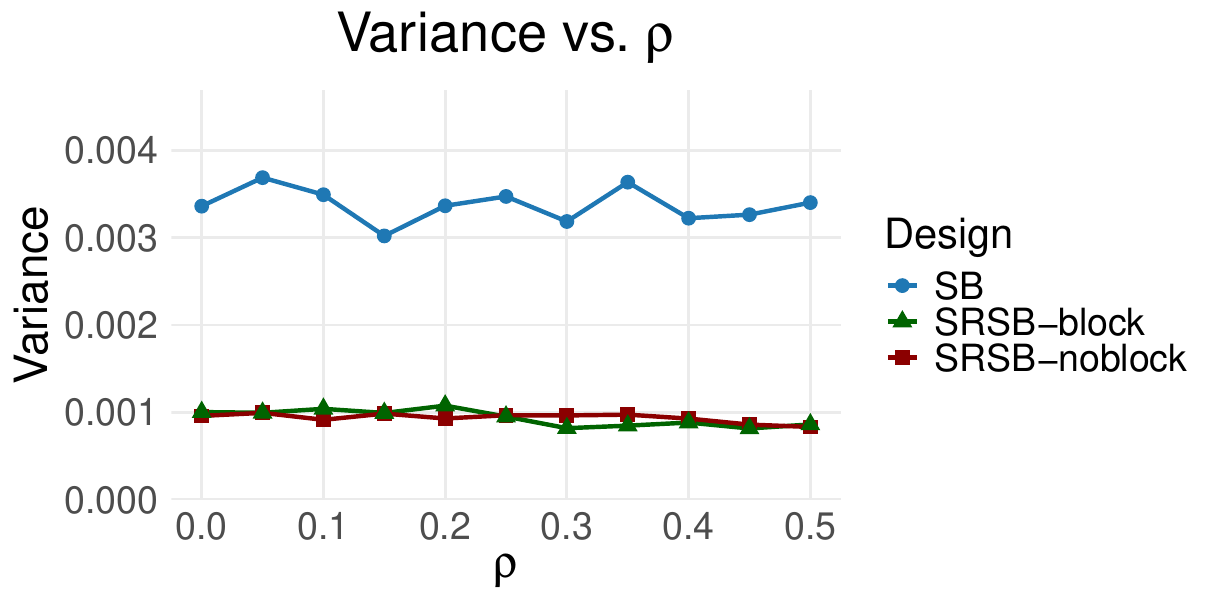}
	\end{minipage}}
    \subfigure[RMSE]{
		\begin{minipage}[t]{0.48\linewidth}
			\centering
			\includegraphics[width=2.5in]{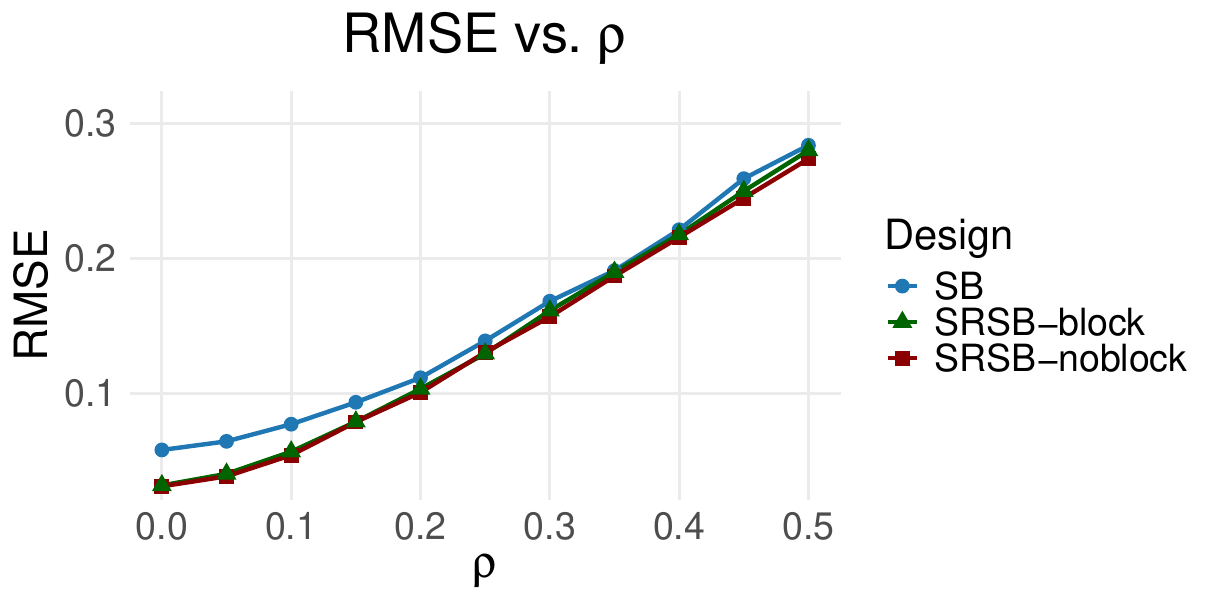}
	\end{minipage}}
	\centering
    \caption{Performance of the estimator across different experimental designs. The potential outcomes are generated from a Markovian carryover model with latent states based on the Penn World Table. See Section~\ref{sec:gdp-markov} for further details.}
	\label{fig:simu-gdp-mdp}
\end{figure}

\section{Discussion}

This paper studies adaptive design and inference for multi-unit switchback experiments, where a fixed set of operational units (e.g., geos) repeatedly switch between treatment and control over time and outcomes/covariates are revealed sequentially. Under a design-based finite-population framework, we propose an SRSB design that adaptively balances observed outcomes and covariates at each period, improving precision when these variables are predictive. In settings without carryover, we develop two inference procedures: (i) randomization inference under a sharp null, valid for finite $(N,T)$, and (ii) asymptotic inference as $T\to\infty$ via a martingale CLT, allowing $N$ to be fixed or to grow with $T$. We then extend SRSB to settings with first-order carryover effects, introducing a blocked SRSB variant that stratifies by the previous assignment and rerandomizes within blocks to produce stable, comparable ``stay'' groups for estimating carryover-aware treatment effects; asymptotic normality is established using mixingale-style arguments. Extensive simulations, including semi-synthetic GDP-based experiments and MDP-style carryover dynamics, demonstrate when SRSB reduces variance and RMSE and how its robustness depends on the strength of carryover and the effect size.

There are several natural directions for future work. In this paper, we focus primarily on settings with first-order (or more generally, finite-order) carryover effects. In practice, however, carryover may be infinite-order \citep{bojinov2019time, bojinov2021panel, shi2023dynamic,hu2024switchbackexperimentsgeometricmixing,johari2025estimation, missault2025carryover}, in which case finite-order approximations can introduce non-negligible bias, as illustrated in our simulations. Another extension is to relax the no-spillover assumption (Assumption~\ref{assumption:no-spillover}) and allow for interference across units \citep{jia2025clusteredswitchbackdesignsexperimentation, masoero2026multiple}. It would be interesting to develop SRSB-type designs and accompanying theory for multi-unit switchback experiments with infinite-order carryover and/or spillover effects. Moreover, the prediction-based variance estimator in Section~\ref{sec:SRSB-carryover} is conservative; developing sharper variance estimators could yield tighter Wald-style confidence intervals and more informative inference.

\newpage
	\bibliographystyle{plainnat}
	\bibliography{references}

\newpage

\appendix

\section{Proof of Theorem \ref{thm:CLT-simple}}

\begin{proof}
    Conditioning on $\mathcal{F}_{t-1}$, since the numbers of treated and control units are equal and the covariate imbalance measure is Mahalanobis distance, by \citet{morgan2012rerandomization}[Theorem 2.1], 
    \[
    \E[\hat{\tau}_t \mid \mathcal{F}_{t-1}]=\tau_t
    \]
    and $\hat{\tau}_t-\tau_t$ is a martingale difference sequence. By Martingale Central Limit Theorem \citep{brown1971martingale}, we have
    \[
    \frac{\sum_{t=1}^T(\hat{\tau}_t-\tau_t)}{S_T}\stackrel{d}{\rightarrow} N(0,1).
    \]
    The sufficient condition for Lindeberg's condition follows since when $|Y_{i,t}| \leq M$, we have $|\hat{\tau}_t-\tau_t|\leq 2M$. Hence as $T \rightarrow \infty$, there exists $T_0 \in \mathbb{N}_+$ such that $ S_T > 2M/\epsilon$ for all $T \geq T_0$, then we have
    \[
    \frac{1}{S_T^2} \sum_{t=1}^T \E \left[(\hat{\tau}_t-\tau_t)^2 I(|\hat{\tau}_t-\tau_t| \geq \epsilon S_T) \right] = 0
    \]
    when $T$ is sufficiently large.
\end{proof}

\section{Proof of Theorem \ref{thm:CLT-carryover}}

\begin{proof}
The proof relies on the following more general lemma.
\begin{lemma}\label{lemma:CLT-carryover}
Suppose the sequence of standardized squared increments
\[
\left\{\frac{(\hat{\tau}_t-\tau_t)^2}{\operatorname{Var}(\hat{\tau}_t)} : 2\le t\le T\right\}
\]
is uniformly integrable.  
Assume there exists a subsequence of time indices
\[
1 = t_1 < t_2 < \cdots < t_J < t_{J+1}=T+1,
\qquad J = J_T \to \infty,
\]
satisfying the following conditions:

\begin{enumerate}
    \item
    For all $j \ge 2$,
    \[
    t_j - t_{j-1} \ge 2.
    \]

    \item
    Define the blockwise predictable variance
    \[
    V_J^2 =  \sum_{j=1}^J \E \left [ \left(\sum_{t=t_j+1}^{t_{j+1}-1} (\hat{\tau}_t -\tau_t) \right)^2 \middle | \mathcal{F}_{t_j-1} \right].
    \]
    Assume that
    \[
    V_J^2/\E[V_J^2] \xrightarrow{P} 1
    \]
    as $T \rightarrow \infty$.

    \item

    Let
    \[
    a_t^2
    \;=\;
    \frac{\|\hat{\tau}_t - \tau_t\|_2^2}{S_T^2}
    =
    \frac{\operatorname{Var}(\hat{\tau}_t)}{S_T^2}.
    \]

    Assume:
    \begin{equation}\label{eq:var-assumption}
        \sum_{j=2}^J a_{t_j}^2 \rightarrow 0,\, \sum_{t=2}^T a_t^2 = O(1),\, \max_{1 \leq j \leq J} \sum_{t=t_{j}+1}^{t_{j+1}-1}a_t^2 \rightarrow 0.
    \end{equation}
    
\end{enumerate}

Then we have
\[
\frac{(T-1) (\hat{\tau}-\bar{\tau})}{S_T}
\;\xrightarrow{d}\;
N(0,1).
\]

\end{lemma}
Most conditions are similar to Theorem \ref{thm:CLT-carryover} in the main text, except the third condition. Equation \eqref{eq:var-assumption} ensures the contribution from the first time period of each block is negligible, and that the sum of the individual variances is of the same order as $S_T^2$, while the variance contribution from each block remains asymptotically negligible.

    We only need to verify Assumption 3 in Lemma \ref{lemma:CLT-carryover}. 
    \[
    \sum_{j=2}^J a_{t_j}^2 = \, \sum_{j=2}^{J_T} a_{(i-1)b_T +1}^2
    \leq J_T \max_{2 \leq t \leq T} a_t^2
     = O(1/b_T),
    \]
    where we use $b_T J_T \sim T$ and \eqref{eq:bd-var}.
    
    \[
    \sum_{t=2}^T a_t^2  \leq T \max_{2 \leq t \leq T} a_t^2  = O(1).
    \]
    \[
    \max_{1 \leq j \leq J} \sum_{t=t_{j}+1}^{t_{j+1}-1}a_t^2 = \max_{1 \leq j \leq J} \sum_{(j-1)b_T+2}^{jb_T}a_t^2 \leq b_T \max_{2 \leq t \leq T} a_t^2 =o(1).
    \]
\end{proof}




\section{Proof of Lemma \ref{lemma:CLT-carryover}}

\begin{proof}

We first introduce the notion of mixingales \citep{mcleish1977invariance,de1997central, davidson1992central}.
\begin{definition}[Mixingales]\label{def:mixingale}
    A triangular array $\{X_{nt}, \mathcal{F}_{nt}\}$ is called an $L_2$-mixingale if for $m \geq 0$
    \[
    \begin{aligned}
        \left\|E\left(X_{n t} \mid \mathcal{F}_{n, t-m}\right)\right\|_2 \leq &\, a_{n t} \psi_m, \\
        \left\|X_{n t}-E\left(X_{n t} \mid \mathcal{F}_{n, t+m}\right)\right\|_2 \leq &\, a_{n t} \psi_{m+1},
    \end{aligned}
    \]
    and $\psi_m \rightarrow 0$ as $m \rightarrow \infty$.
\end{definition}

Let $X_{Tt}=\frac{\hat{\tau}_t - \tau_t}{S_T}$ with $S_T^2 = \operatorname{Var} \left( \sum_{t=1}^T\hat{\tau}_t\right)=\sum_{t=1}^T \E[(\hat{\tau}_t-\tau_t)^2] + 2 \sum_{t=1}^{T-1} \E[(\hat{\tau}_t-\tau_t)(\hat{\tau}_{t+1}-\tau_{t+1})] $ and $\mathcal{F}_{Tt}=\sigma(\bW_1, \dots, \bW_t)$. For notation simplicity we will write $X_t = X_{Tt}, \mathcal{F}_{Tt}=\mathcal{F}_t$ and all the arguments below hold for triangular array. The second condition in Definition \ref{def:mixingale} is automatically satisfied because $X_t \in \mathcal{F}_t$. Now check the first condition. If $m \geq 2$ we have $\E[X_{t} \mid \mathcal{F}_{t-m}]=0$ and thus $\psi_m = 0$.

When $m = 0$ the first condition is 
\[
\|\E[X_t \mid \mathcal{F}_t]\|_2 \leq a_t \psi_0, 
\]
here $\E[X_t \mid \mathcal{F}_t] = X_t$.
When $m = 1$ the first condition is 
\[
\|\E[X_t \mid \mathcal{F}_{t-1}]\|_2 \leq a_t \psi_1.
\]
By Jenson's inequality we have
\[
\|\E[X_t \mid \mathcal{F}_{t-1}]\|_2 = \sqrt{\E [\E^2(X_t \mid \mathcal{F}_{t-1}) ]} \leq \sqrt{\E[X_t^2]} = \|X_t\|_2.
\]
Hence we have shown $\{X_t, \mathcal{F}_t\}$ is an $L_2$-mixingale with. $a_t = \|X_t\|_2=\|\hat{\tau}_t -\tau_t\|_2/S_T, \psi_0 = \psi_1 = 1$ and $\psi_m = 0, m\geq 2$.

Following equation \eqref{eq:carryover-decompose}, the self-normalized sum $\sum_{t=2}^T (\hat{\tau}_t - \tau_t)/S_T$ can be expressed as 
\begin{equation}\label{eq:decomposition-appendix}
    \begin{aligned}
    &\,\frac{\sum_{t=2}^T (\hat{\tau}_t - \tau_t)}{S_T} \\
    = &\, \sum_{t=2}^T X_t \\
    =&\, \sum_{j=1}^{J_T} Z_j + \sum_{j=2}^{J_T} X_{t_j} .
\end{aligned}
\end{equation}
\textbf{Step 1: $\sum_{j=2}^{J_T} X_{t_j}$ is asymptotically negligible.} Note that when $t_j - t_{j-1} \geq 2$, $X_{t_j}$ and $X_{t_{k}}$ are uncorrelated for $j \neq k$. Hence
\[
\begin{aligned}
    &\,\E\left[ \left( \sum_{j=2}^{J_T} X_{t_j} \right)^2\right] \\
    = &\, \sum_{j=2}^{J_T} \E[X_{t_j}^2]\\
    = &\, \sum_{j=2}^{J_T} a_{t_j}^2\\
    = & \, o(1) ,
\end{aligned}
\]
where the last equation follows from Assumption 3 in Theorem \ref{thm:CLT-carryover}. This implies
\[
\sum_{j=2}^{J_T} X_{t_j} \stackrel{P}{\rightarrow} 0.
\]
\textbf{Step 2: $\sum_{j=1}^{J_T} Z_j$ is asymptotically normal.} Since $\E[\hat{\tau}_t \mid \mathcal{F}_{t-2}]=\tau_t$, when $b_T \geq 2$ $\E[Z_j \mid \mathcal{F}_{t_j-1}]=0$ and thus the triangular array $\{Z_j, \mathcal{F}_{t_{j+1}-1}, 1 \leq j \leq J_T\}$ is a martingale difference sequence. Let $S_T^{*2} = \sum_{j=1}^{J_T} \E \left[ Z_j^2\right]$ be the variance of $\sum_{j=1}^{J_T} Z_j$. We first show $S_T^{*2} \rightarrow 1$ as $T \rightarrow \infty$. Note that $\E\left[ \left( \sum_{t=2}^T X_t\right)^2 \right]=1$, recall the decomposition in equation \eqref{eq:decomposition-appendix}, we have
\[
\begin{aligned}
    \E\left[ \left( \sum_{t=2}^T X_t\right)^2 \right]
    = &\, \E\left[ \left( \sum_{j=1}^{J_T} Z_j\right)^2 \right] + \E\left[ \left( \sum_{j=2}^{J_T} X_{t_j} \right)^2\right]+2 \E\left[ \left( \sum_{j=1}^{J_T} Z_j\right) \left( \sum_{j=2}^{J_T} X_{t_j} \right)\right].\\
\end{aligned}
\]
By the first step, we have
\[
\E\left[ \left( \sum_{j=2}^{J_T} X_{t_j } \right)^2\right] = o(1).
\]
By Cauchy-Schwarz inequality we have
\[
\left|\E\left[ \left( \sum_{j=1}^{J_T} Z_j\right) \left( \sum_{j=2}^{J_T} X_{t_j} \right)\right] \right| \leq \sqrt{\E\left[ \left( \sum_{j=1}^{J_T} Z_j\right)^2 \right] \E\left[ \left( \sum_{j=2}^{J_T} X_{t_j} \right)^2\right]}=o(S_T^*),
\]
Thus we have
\[
S_T^{*2} + o(S_T^*) +o(1)=1,
\]
which implies $S_T^{*2} \rightarrow 1$, or equivalently,
\[
\E[V_J^2]/S_T^2 \rightarrow 1.
\]
This shows that $\sum_j Z_j$ dominates in equation \eqref{eq:decomposition-appendix} and 
\begin{equation}\label{eq:stable-var-Z}
    \sum_{j=1}^J \E[Z_j^2 \mid \mathcal{F}_{t_j-1}]= V_J^2/S_T^2 \stackrel{P}{\rightarrow} 1
\end{equation}
by Assumption 2.

We then verify the Lindeberg's condition for $\{Z_j, \mathcal{F}_{t_{j+1}-1}, 1 \leq j \leq J_T\}$:
\[
\lim_{T \rightarrow \infty} \sum_{j=1}^{J_T} \E \left [ Z_j^2 I(|Z_j| \geq \epsilon)\right] =0.
\]
Note that we have shown that $S_T^{*2} \rightarrow 1$ so there is no need to normalize $\sum_{j=1}^{J_T} Z_j$. Define $\nu_{j}^2 = \sum_{t=t_j+1}^{t_{j+1}-1} a_t^2$. Following the proof in \cite{de1997central, davidson1992central}, we have
\[
\begin{aligned}
    &\,\sum_{j=1}^{J_T} \E \left[ Z_j^2 I (|Z_j| \geq \epsilon) \right]\\
    = &\, \sum_{j=1}^{J_T} \E \left[ \frac{Z_j^2}{\nu_j^2} I (|Z_j| \geq \epsilon) \right] \nu_j^2\\
    \leq &\, \max_{1 \leq j \leq J_T} \E \left[ \frac{Z_j^2}{\nu_j^2} I (|Z_j| \geq \epsilon) \right] \sum_{j=1}^{J_T} \nu_j^2\\
    \leq &\, \max_{1 \leq j \leq J_T} \E \left[ \frac{Z_j^2}{\nu_j^2} I (|Z_j|/\nu_j \geq \epsilon /\nu_j) \right] \sum_{j=1}^{J_T} \nu_j^2.
\end{aligned}
\]
By equation \eqref{eq:var-assumption} in Assumption 3 
\[
 \sum_{j=1}^{J_T} \nu_j^2 = O(1),
\]
\[
\max_j \nu_j^2  = o(1).
\]
Hence after we prove $\{Z_j^2/\nu_j^2, 1\leq j \leq J_T\}$ is uniformly integrable, it follows that
\[
\lim_{T \rightarrow \infty} \max_{1 \leq j \leq J_T} \E \left[ \frac{Z_j^2}{\nu_j^2} I (|Z_j|/\nu_j \geq \epsilon /\nu_j) \right] = 0
\]
and Lindeberg's condition holds. The uniform integrability of $\{Z_j^2/\nu_j^2, 1\leq j \leq J_T\}$ then follows from \cite{davidson1992central}[Lemma 3.2] and the discussion after that lemma. The Lindeberg's condition and equation \eqref{eq:stable-var-Z} implies
\[
\sum_{j=1}^{J_T} Z_j \stackrel{d}{\rightarrow } N(0,1).
\]
Combining these two steps, by Slutsky's theorem, we have
\[
\frac{\sum_{t=2}^T (\hat{\tau}_t - \tau_t)}{S_T} \stackrel{d}{\rightarrow } N(0,1).
\]
\end{proof}

\section{Additional Simulation Results}

\subsection{Heterogeneous Individualized Treatment Effects}

In this section, we report additional simulation results under heterogeneous treatment effects with first-order carryover. We generate the autoregressive process $U_{i,t}$ as in Section~\ref{sec:simu-carryover} and define the potential outcomes by
\[
Y_{i,t}(0,0)=U_{i,t}, \, Y_{i,t}(0,1)=U_{i,t}+2B_{i,t}^{(1)},
\]
\[
Y_{i,t}(1,0)=U_{i,t}+B_{i,t}^{(2)}, \, Y_{i,t}(1,1)=U_{i,t}+7B_{i,t}^{(3)},
\]
where $B_{i,t}^{(1)}, B_{i,t}^{(2)}, B_{i,t}^{(3)}$ are i.i.d. Bernoulli random variables with mean 0.5. The estimand of interest is
\[
\bar{\tau} = \frac{1}{N(T-1)}\sum_{t=2}^T \sum_{i=1}^N \{Y_{i,t}(1,1)-Y_{i,t}(0,0)\} \approx3.5.
\]
We consider the same simulation settings as in Section~\ref{sec:simu-carryover} and report the results in Figure~\ref{fig:simu-carryover2}.

\begin{figure}[t]
	\centering
	\subfigure[$T=40, N \in \{200, 300, \dots, 1000\}$]{
		\begin{minipage}[t]{0.48\linewidth}
			\centering
			\includegraphics[width=2.5in]{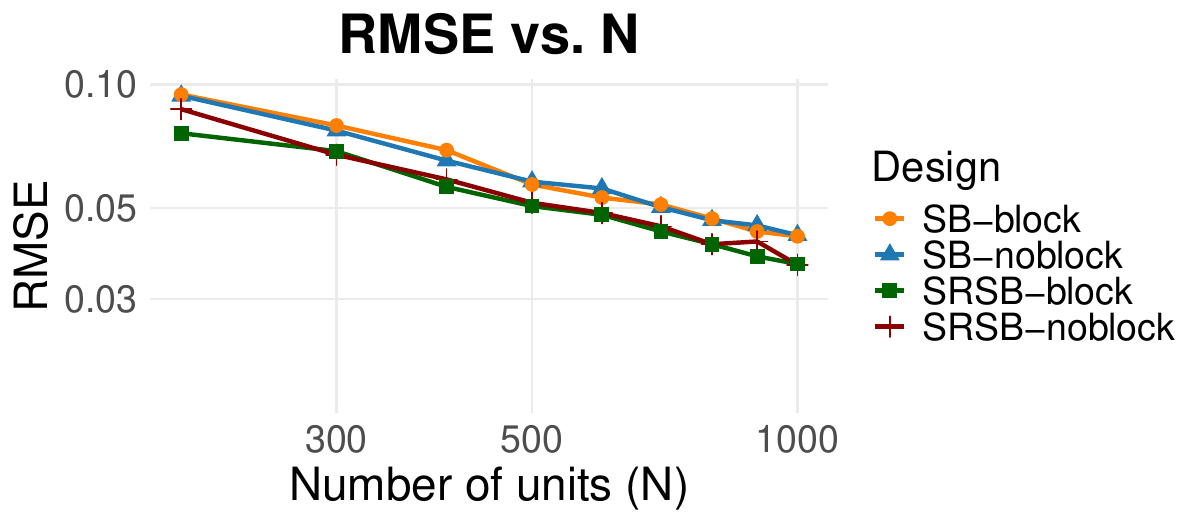}
	\end{minipage}}
	\subfigure[$N=600, T \in \{10, 20, \dots,100\}$]{
		\begin{minipage}[t]{0.48\linewidth}
			\centering
			\includegraphics[width=2.5in]{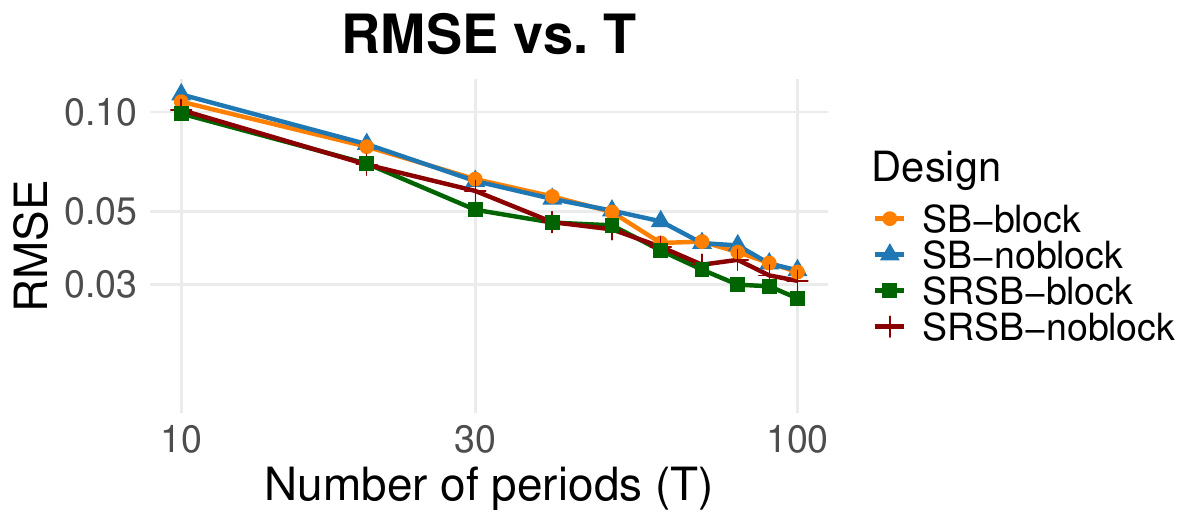}
	\end{minipage}}
	\centering
	\caption{Comparison of RMSE under different experimental designs when individual treatment effects are heterogeneous.}
	\label{fig:simu-carryover2}
\end{figure}

Consistent with the results in the main text, the proposed SRSB design outperforms the switchback designs based on complete randomization (both unblocked and blocked). In particular, blocked SRSB typically achieves lower estimation error than its non-blocked counterpart, highlighting its advantage in the presence of first-order carryover effects.

\subsection{Variance Estimation}

In this section, we conduct simulations to illustrate the performance of our conservative prediction-based variance estimator~\eqref{eq:var-est}. We follow the setup in Section~\ref{sec:gdp-carryover}, comparing different designs using the semi-synthetic data constructed from the Penn World Table in the presence of first-order carryover effects. In each replication, we simulate the data and estimate the treatment effect under each design. We then compute the variance estimator in~\eqref{eq:var-est} with block size $b_T=8$, where the prediction based on $\mathcal{F}_{t_j-1}$ is defined as a scaled average of past period estimates:
\[
M_j 
= \frac{t_{j+1} - t_j - 1}{t_j - 1} 
  \sum_{t=1}^{t_j - 1} \hat{\tau}_t.
\]
Using this variance estimator, we construct the Wald-style confidence interval and evaluate its coverage and length. We repeat this procedure for $M=500$ replications and report the results in Figure~\ref{fig:simu-var}.

\begin{figure}[t]
	\centering
	\subfigure[Coverage Probability]{
		\begin{minipage}[t]{0.48\linewidth}
			\centering
			\includegraphics[width=2.5in]{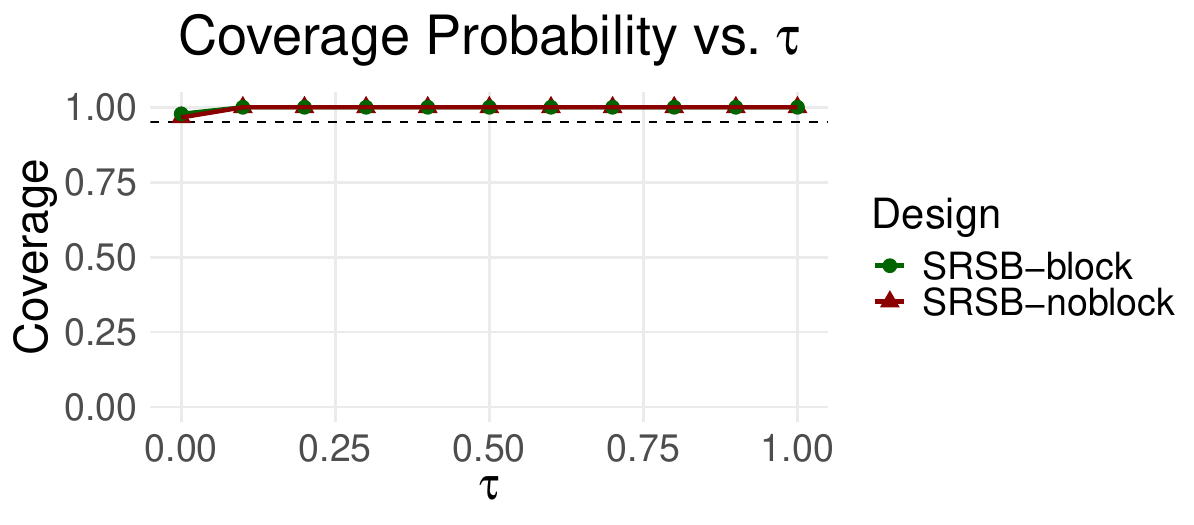}
	\end{minipage}}
	\subfigure[Length of confidence intervals]{
		\begin{minipage}[t]{0.48\linewidth}
			\centering
			\includegraphics[width=2.5in]{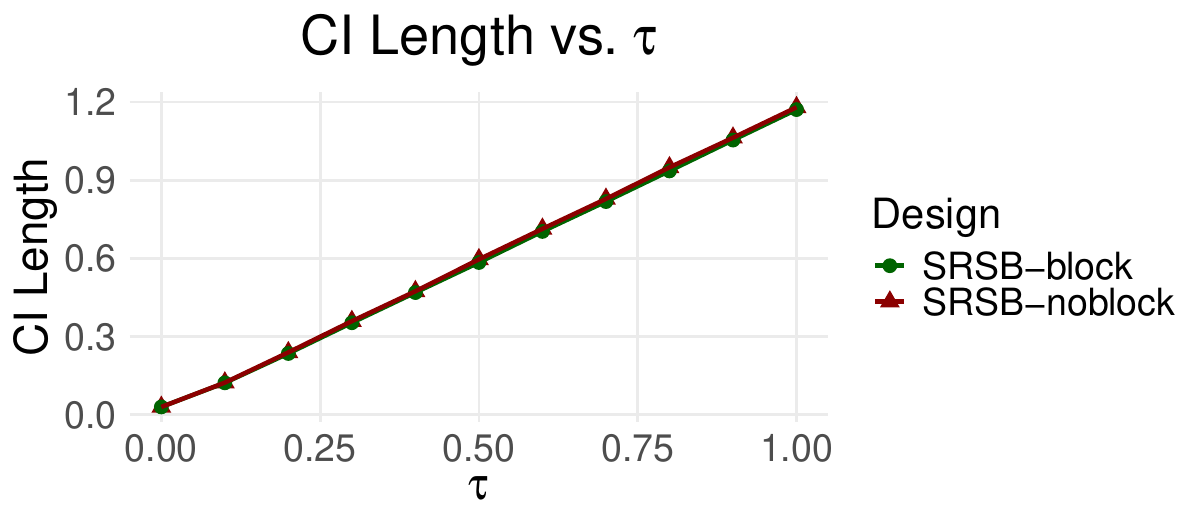}
	\end{minipage}}
	\centering
	\caption{Properties of Wald-style confidence intervals based on the conservative variance estimator}
	\label{fig:simu-var}
\end{figure}

As shown in Figure~\ref{fig:simu-var}(a), the variance estimator is conservative, so the empirical coverage probability exceeds the nominal level of 0.95. Figure~\ref{fig:simu-var}(b) shows that the average CI length increases with $\tau$ (recall that the treatment effect of interest is $2\tau$). Although the variance estimator is conservative, the resulting confidence intervals still have adequate power to detect non-zero effects when $\tau\neq 0$. For example, when $\tau=0.5$ (so the treatment effect is $1$), an interval centered at an unbiased point estimate with length $0.6$ is approximately $[0.7,1.3]$, which excludes $0$ and would reject the null of zero treatment effect. Overall, the conservative variance estimator \eqref{eq:var-est} can still yield valid and informative inference in this setting.

\end{document}